\documentclass[11pt]{article}
\usepackage[margin=1in]{geometry}

\usepackage{natbib}
\setcitestyle{authoryear}

\usepackage{microtype}
\usepackage[utf8]{inputenc}
\usepackage{graphicx,paralist}
\usepackage{amsmath, amsthm, amssymb}
\usepackage{dsfont}
\usepackage{comment}
\usepackage{hyperref}
\hypersetup{
	colorlinks=true,
	citecolor = blue!80!black,
	linkcolor = red!80!black
}

\newcommand{\email}[1]{\href{mailto:#1}{\nolinkurl{#1}}}
\usepackage{xcolor}
\usepackage{thmtools}
\usepackage{tikz}
\usepackage{pgfplots}
\usepackage{varwidth}
\pgfplotsset{compat=1.10}
\usepgfplotslibrary{fillbetween}
\usetikzlibrary{backgrounds}
\usetikzlibrary{patterns}
\usepackage{enumitem}
\usepackage{array}
\usepackage{multirow}
\usepackage{booktabs}
\usepackage{stmaryrd}
\usepackage{ragged2e}
\usepackage[font=small]{caption}

\usepackage[noabbrev,capitalise,nameinlink]{cleveref}

\newtheorem{theorem}{Theorem}
\newtheorem{proposition}{Proposition}
\newtheorem{lemma}{Lemma}
\newtheorem{observation}{Observation}
\newtheorem{corollary}{Corollary}

\newcommand{\R}{\mathbb{R}}
\newcommand{\E}{\mathbb{E}}
\newcommand{\PP}{\mathbb{P}}
\newcommand{\OPT}{\ensuremath{\mathrm{OPT}}}
\newcommand{\ALG}{\ensuremath{\mathrm{ALG}}}
\newcommand{\ceil}[1]{\left\lceil #1 \right\rceil}
\crefname{ineq}{Inequality}{inequalities}
\creflabelformat{ineq}{#2{\upshape(#1)}#3}

\begin{document}

\title{Trading Prophets\thanks{This work was initiated during the Fall 2021 Virtual Chair Semester on Prophet Inequalities~\citep{CorreaIH2022}. 
We would like to thank the organizers and the participants of the semester for their feedback.
Kevin Schewior is partially supported by the Independent Research Fund Denmark, Natural Sciences, under Grant No.\ DFF-0135-00018B. MohammadTaghi Hajiaghayi and Jan Olkowski are partially supported by DARPA QuICC, NSF AF:Small 2218678, and NSF AF:Small 2114269.}}

\author{Jos\'e Correa\thanks{Universidad de Chile. Email: \email{correa@uchile.cl}, \email{andres.cristi@ing.uchile.cl}} \and Andr\'es Cristi\footnotemark[2] \and Paul D\"utting\thanks{Google Research. Email: \email{duetting@google.com}} \and Mohammad Hajiaghayi\thanks{University of Maryland. Email: \email{hajiagha@umd.edu}, \email{olkowski@umd.edu}} \and Jan Olkowski\footnotemark[4] \and Kevin Schewior\thanks{University of Southern Denmark. Email: \email{kevs@sdu.dk}}}

\date{}

\maketitle 

\begin{abstract}
In this work we initiate the study of buy-and-sell prophet inequalities. We start by considering what is arguably the most fundamental setting. In this setting the online algorithm observes a sequence of prices one after the other. At each time step, the online algorithm can decide to buy and pay the current price if it does not hold the item already; or it can decide to sell and collect the current price as a reward if it holds the item.

We show that for i.i.d.~prices a single-threshold online algorithm achieves at least $1/2$ of the expected profit of the optimal offline algorithm and we prove that this is optimal. For non-i.i.d.~prices in random order, where prices are no longer independent, we give a single-threshold online algorithm that achieves at least a $1/16$ fraction of the expected profit of the optimal offline algorithm. We also show that for this setting no online algorithm can yield a better than $1/3$ approximation, and thus establish a formal separation from the i.i.d.~case. On the other hand, we present a threshold-based online algorithm for this setting that yields a $1/2-o(1)$ approximation. For non-i.i.d.~prices no approximation is possible.

We use the results for these base cases to solve a variety of more complex settings. For instance, we show a $1/2-o(1)$ approximation for settings where prices are affiliated and the online algorithm has only access to a single sample. We also extend our upper and lower bounds for the single item case to $k$ items, and thus in particular show that it is impossible to achieve $1-o(1)$ approximations. For the budgeted version, where fractions of an item can be bought, and gains can be reinvested, we show a constant-factor approximation to the optimal offline algorithm's growth rate. In a setting with $k$ item types and price streams, we achieve a $\Omega(1/k)$ approximation for the unit-capacity case, which is optimal.
\end{abstract}

\section{Introduction}

Consider a trader, let's call her Alice, buying and selling a certain good or commodity (e.g., paintings, used cars, barrels of oil, etc.). Suppose Alice has some limit on how many units she can store. 
Consider the situation where Alice faces a sequence of potential buyers/sellers that arrive online over
time and approach her with offers, and wants to maximize her profit. Of course, she can only sell a unit of the good if she
has it, and she can only buy a unit if she has enough place to store it. How should she
proceed?

Or consider Bob, who wants to invest $\$1000$ in Bitcoin. Suppose Bob seeks to buy and sell coins (which is possible in essentially arbitrary fractions) to make as much money as he can. He is willing to reinvest any gains but he doesn't want to put in extra money. So Bob's constraint is his current budget, and not his capacity. Suppose that on any given day, there is a single market price at which Bob can buy and sell coins. He can buy as many coins as he can afford, and he can sell some or all of his current shares. How should he proceed?

Of course, there is a whole range of situations like this, and different situations will justify different modelling assumptions \cite[e.g.][]{CharnesEtAl66, Osborne59, MyersonS83}. For example, in Alice's case, one could model buyer/seller valuations as independent draws from distinct distributions that arrive in random order. Another plausible assumption could be an affiliated-valuations model, where buyers/sellers share a random base value but their individual preferences are expressed by independent increments.

In this work, we initiate the study of this type of problems from a prophet-inequality perspective. That is, we assume that the trader has some prior (distributional) knowledge about the sequence of prices, and compare the expected performance of a given online algorithm to the expected performance of the best offline algorithm. 
A fundamental challenge of our buy-and-sell prophet-inequality problem and important departure from the classic prophet inequality problem is the mixed-sign objective.
We start by considering what is arguably the most basic variant of a buy-and-sell problem, where the online algorithm can hold up to one indivisible item, she is not budget-constrained, and prices are either i.i.d.\ random variables, non-i.i.d.~presented in random order, or non-i.i.d.~presented in adversarial order. We then show how to use the results for these base cases to solve a variety of more complex settings, including settings with limited information about the distribution of prices, settings with additional forms of correlation among prices
(such as the aforementioned affiliated-valuations model),
settings with more than one item, settings with multiple types of items, settings with divisible items, and settings with budgets.
We also observe that in the arbitrarily correlated case and in the case where prices form a martingale, no approximation is possible.

\subsection{The Basic Trading Prophet Problem}
At the heart of our work is the following basic  \emph{trading prophet} problem. A ``gambler'' (online algorithm) observes a sequence of $n$ prices. The prices are generated by a stochastic process that is known to the agent, but the agent only observes the realized prices one-by-one in an online fashion. 
In the simplest and most fundamental version of the problem that we consider, the agent trades an indivisible good and can hold at most one copy of the good.
In this case whenever the agent does not have the item, she can either buy it at the current price $p_j$ (at a reward of $-p_j$) or she can skip this price. Similarly, when she holds the item, she can either sell it at the current price $p_j$ (collecting a reward of $p_j$) or she can hold on to the item. 

The agent can base her buy/sell decisions on knowledge about the stochastic process, the current state (e.g., whether she currently holds the item or not), and the history of prices. 
We seek to compare the expected reward that the agent can achieve this way to the expected reward of an all-knowing ``prophet'' (offline algorithm), who knows the entire sequence of prices in advance and can make optimal buy and sell decisions. 

Our goal are worst-case approximation guarantees (``prophet inequalities'') that state that in the worst-case over all stochastic processes from a certain class, the expected reward of the online algorithm is at least an $\alpha \in (0,1]$ fraction of the expected reward of the optimal offline algorithm.

\subsection{Results for the Basic Trading Prophet Problem}

We show a number of tight and near-tight bounds for the fundamental version of the trading-prophet problem in which the agent can hold up to one item. 
Our results for the base case are summarized in~\cref{tab:results}. We discuss a number of extensions, including additional results for correlated prices (beyond those discussed below) and budgets in~\cref{sec:extensions-overview}.

\smallskip

\paragraph{Results for I.I.D.~Prices.} We start with the case where the prices $p_j$ are i.i.d.~draws from~$F$. We show that the worst-case optimal online policy is a simple threshold policy and achieves a $1/2$-approximation. Note that we are assuming continuous distributions without loss of generality.

\medskip

\noindent{\bf Main Result 1.} If the $p_i$ are i.i.d.~draws from $F$, then buying whenever the price is below the median and selling otherwise yields a $1/2$-approximation; and no online policy can do better.

\medskip

We formally prove this result in~\cref{sec:iid}, but we give a high-level overview here: We argue that the optimal offline strategy buys whenever the price is a local minimum and sells at local maxima, and use this to show that the expected reward of the prophet is exactly $(n-1)/2$ times the expected absolute difference $\E[|X_i-X_j|]$ between any two draws $X_i, X_j, i \neq j$ from $F$, say between $X_1$ and $X_2$. A simple application of the triangle inequality in combination with the fact that we are looking at i.i.d.~random variables then leads to an upper bound of $(n-1) \cdot \E[|X_1 - T|]$ for any $T$. On the other hand, for $T$ set to the median of $F$, we show the agent has the item at any given step with probability $1/2$. This means that for any intermediate step $1 < i < n$, with half probability whenever $X_i < T$ the agent can buy and with half probability whenever $X_i \geq T$ the agent can sell. So the expected reward from an intermediate step is $1/2 \cdot (\E[X_1 \cdot \mathds{1}_{X_1 \geq T}] - \E[X_1 \cdot \mathds{1}_{X_1 < T}])$. Now since $T$ is the median, $\E[T \cdot \mathds{1}_{X_1 \geq T}] = \E[T \cdot \mathds{1}_{X_1 < T}]$. So we can add and subtract this term from the previous formula, and obtain that the expected reward from an intermediate step is $1/2 \cdot \E[|X_1-T|]$.
Together with a careful analysis of the boundary cases this shows that the expected reward achieved by the agent over all steps is exactly equal to $(n-1)/2 \cdot \E[|X_1 - T|]$.

To establish that no online policy can achieve a better approximation guarantee than $\frac{1}{2}$, consider constructing an i.i.d. instance with $n = 2$ random variables $X_i \sim F$ for $i \in \{1,2\}$, where $F$ is supported on three values: $0$, $\frac{1}{2}$, and $1$. By appropriately selecting the probability distribution for these values, it is possible to ensure that the best online policy only buys in period 1 when it observes the lowest possible value. Namely, our distribution is such that the expected period 2 price is not larger than the middle value. In contrast, the prophet will occasionally buy when it observes the middle value in period 1 and when the period-2 price is the highest possible value. This idea can also be generalized to $n \ge 2$, demonstrating that increasing $n$ does not lead to better online policies.

\begin{table}[t]
    \centering
    \begin{tabular}{r>{\centering}m{3cm}>{\centering}m{3cm}>{\centering\arraybackslash}m{3cm}}
\toprule         \multirow{2}{*}{}&\multirow{2}{*}{i.i.d.\ prices}&\multicolumn{2}{c}{non-i.i.d.~prices}\\\cmidrule{3-4}
&&random order&worst order\\
         \midrule
        upper bound & $2$\hspace{3cm}{\scriptsize (\cref{thm:iid_half_approx})} & $16$ ($2$ as $n\rightarrow\infty$)\hspace{3cm}{\scriptsize (\cref{thm:random-order,thm:random-order-v2})} & --\\[.5cm]
        lower bound & $2$\hspace{3cm}{\scriptsize (\cref{prop:iid_lb})} & $3$ \hspace{3cm}{\scriptsize (\cref{prop:rdm_order_lb})} & $\infty$\hspace{3cm}{\scriptsize (simple example)}\\
        \bottomrule
    \end{tabular}
    \caption{The achievable approximation ratios in the basic versions of the trading-prophet problem.}
    \label{tab:results}
\end{table}

\paragraph{Results for Non-I.I.D.~Prices, Random Order.} 
Now assume that prices $p_j$ are first drawn independently from non-identical distributions $F_j$, and that they are then presented to the algorithms in random order.
Note that when prices are generated this way, then prices observed on different days are no longer independent of each other. 
We show that despite this, it's still possible to achieve a $1/16$ approximation guarantee with a 
threshold policy, which sets a single non-adaptive threshold.
We also show that this case is strictly harder than the i.i.d.~case by showing that no online algorithm (whether single threshold or not) can achieve a better than $1/3$ approximation. Interestingly, this impossibility applies even if the order of random variables, not their realizations, is revealed to the algorithm before the sequence starts. We also present a threshold-based online algorithm that achieves a $1/2-\frac{c}{n}$-approximation, for some constant $c$. 

\medskip

\noindent{\bf Main Result 2.} For non-i.i.d.~prices $p_j \sim F_j$ presented in random order there exists a threshold $T$ (different from the median of the mixture distribution) such that the corresponding threshold policy that buys below $T$ and sells above $T$ achieves a $1/16$-approximation, and no online policy can achieve a better than $1/3$-approximation.

\medskip

\noindent{\bf Main Result 3.} For non-i.i.d.~prices $p_j \sim F_j$ presented in random order, setting the threshold $T$ to the median of the mixture distribution yields a $(1/2-\frac{c}{n})$-approximation, for some constant $c$.

\medskip

We give formal proofs in~\cref{sec:rand-order}, but we again give some overview here. The core of our argument is a reduction to a two-period problem with two (random) random variables $X_{\sigma(1)}$ and $X_{\sigma(2)}$, where $\sigma$ is a uniform random permutation. Generalizing our approach for the i.i.d.~case we show that the expected reward of the prophet is $(n-1)/2 \cdot \E[|X_{\sigma(1)} - X_{\sigma(2)}|]$, while the expected reward of a threshold policy with threshold $T$ is equal to
\[
(n-1)/2 \cdot \E[|X_{\sigma(1)} - X_{\sigma(2)}| \cdot \mathds{1}_{T \in [ \min(X_{\sigma(1)},X_{\sigma(2)}),\max(X_{\sigma(1)},X_{\sigma(2)})]}].
\]

The main difficulty for relating the two quantities is that, unlike in the i.i.d.~case where the comparison was via $X_1, X_2$ and these were independent random variables, now we are looking at $X_{\sigma(1)}, X_{\sigma(2)}$, which are no longer independent. 
For the universal $1/16$-approximation we address this correlation by arguing that we can approximate the two quantities of interest up to constant factors with two independent random variables $X_{a'}, X_{b'}$. For this we consider all possible ways of splitting the random variables into two equal halves and drawing $X_{a'}, X_{b'}$ from the two halves, and argue that there must be a split that leads to the required property.
For the asymptotic approximation guarantee of $(1/2-o(1))$ we formalize the intuitive claim that, the larger $n$, the less correlated the two random variables $X_{\sigma(1)}, X_{\sigma(2)}$ should be.

To show that no online policy can obtain a better than $1/3$-approximation, we consider the following instance with $n = 2$ random variables. Let $M$ be a large constant. The first random variable is equal to $M+2$ with probability $M/(M+2)$ and zero otherwise. The second random variable is $M$ with probability $M/(M+2)$ and $2M+2$ otherwise. With probability $1/2$ we first see $X_1$ and then $X_2$ and with probability $1/2$ we first see $X_2$ and then $X_1$. The best online policy (found by backward induction) can be shown to buy only when $X_{\sigma(1)} = 0$ which leads to an expected profit of $1$. The prophet in turn buys and sells except when $X_{\sigma(1)} = 2M+2$ or $X_{\sigma(2)} = 0$, which leads to an expected value of $3- O(1/M)$.

\paragraph{Results for Non-I.I.D.~Prices, Adversarial Order.}
If prices are distributed non-identically and presented in arbitrary order, then no constant approximation factor is possible. Somewhat reminiscent of the classic lower-bound example for the standard prophet inequality, for some $\varepsilon>0$, let $n=2$ and $$X_1=1;\; X_2=\begin{cases}\frac1\varepsilon&\text{w.p. }\varepsilon,\\0&\text{w.p. }1-\varepsilon.\end{cases}$$ An optimal algorithm is indifferent between buying $X_1$ and not doing so, therefore obtaining an expected value of $0$. The prophet, however, buys $X_1$ precisely when $X_2=1/\varepsilon$ (which happens with probability $\varepsilon$), therefore obtaining an expected value of $\varepsilon\cdot(1/\varepsilon-1)$.
Note how the exact same example can be interpreted as a joint distribution over pairs of prices, namely $(1,\frac{1}{\varepsilon})$ with probability $\varepsilon$ and $(1,0)$ with probability $1-\varepsilon$, and so it also shows that with arbitrarily correlated prices no approximation is possible.

\paragraph{(Non-)Connection to Bilateral Trade.} Our analysis---especially the random-order case---reveals an interesting (non-)connection to bilateral trade. Namely, the two-period problem that we reduce to can be seen as the problem of maximizing \emph{gains from trade} when the item is allocated to one of the two parties at random.

For the classic variant of this problem where the item is allocated to a fixed side of the market, classic work by \citet{MyersonS83} showed that the ``second best'' (i.e., the gains from trade achievable with a truthful mechanism) is generally strictly smaller than the ``first best'' (theoretical optimum). \citet{McAfee08} showed that a fixed price mechanism achieves a $1/2$ approximation when the two sides are identically distributed.

For general non-identical distributions it remained open whether a constant-factor approximation is possible. Only very recently, \citet{DengSW21} were able to resolve this question in an affirmative way. They showed that---unlike in our case where no constant-factor approximation is possible in the general case---the classic problem admits a $1/8.23$ approximation.

\subsection{Extensions and Applications}\label{sec:extensions-overview}

Just like the classic single-item prophet inequality \citep{KrengelS77,KrengelS78,SamuelC84}, the basic trading-prophet problem that we introduce in this paper can be used as a tool to solve a number of related problems, and it can be extended in different directions. This includes settings where prices rather then being independent are affiliated, or settings where the trader is subject to a budget constraint. We give an overview here, which is made precise in~\cref{sec:extensions}.

\paragraph{Unknown Distribution and Affiliated Prices.} We show how to handle the version of the i.i.d.\ problem in which the distribution is not known. We give a simple trading strategy that achieves a $(1/2-o(1))$-approximation. The idea is to use a random sample rather than the median as a threshold, and economize on the use of samples using a variation of the fresh-looking samples idea of \citet{CorreaDFS19,CorreaDFSZ21}. The same algorithm works and yields a $(1/2-o(1))$-approximation when prices exhibit a form of correlation known as affiliation in Economics, i.e., there is a distribution $G$ and n distributions $F_1, \ldots, F_n$ and price $p_j = x_j + y$ where $x_j \sim F_j$ and $y \sim G$.

\paragraph{More Than One Item.} We show that unlike in the single-item prophet inequality problem, where better approximation guarantees can be obtained when there are $k$ copies of the good \citep{Alaei14}, both the optimal offline algorithm and the optimal online algorithm for the trading-prophet problem satisfy an ``all or nothing'' property. So they either buy all $k$ copies that they are allowed to buy, or they sell all of the copies that they currently hold. So a simple scaling argument implies that the upper and lower bounds that we identified for the single-item case also apply in the multi-item case.

\paragraph{Budgeted Version with Fractional Purchase and Re-investment of Gains.} We show how our result can be applied in a budgeted-setting similar to that faced by a budget-constrained investor in the stock market. We assume that the agent starts with a budget of $B_1 = 1$ and no goods $S_1 = 0$. In each step $j$ the agent can either buy a fractional amount of the good (at most $B_j/p_j$ units, where $B_j$ is the budget in period $j$) or sell any fraction of the amount $S_j$ of the good that she currently possesses. We show through reduction that it's possible to achieve constant-factor approximations to the growth rate of the optimal offline policy.
 
\paragraph{Multi-Armed Bandit Version.} We present results for a variant of the basic trading-prophet problem in which there are $k$ different types of good.  We model this by running $k$ parallel copies of our base model, so there are $k$ streams of prices, one for each type of good. Both the agent and the prophet can buy any of the items at the current price if they currently don't hold an item, or they can sell the item that they hold at the current price of that item. We use our results for the basic trading-prophet problem to show that it is possible to achieve $O(k)$ approximations, and present an asymptotically tight $\Omega(k)$ lower bound.

\paragraph{Random Walks and Martingales.} A natural extension/variant of our model concerns situations where the prices form a balanced random walk, i.e., $\smash{p_j=p+\sum_{i=1}^{j}x_i}$ where $p$ is some base price and $x_1,\dots,x_n$ are i.i.d.\ variables, each either $-1$ or $1$ with equal probability. Unfortunately, the resulting sequence of prices is a martingale, and it follows from the well-known \emph{optional stopping theorem} that the expected value of any online algorithm is non-positive. Hence, no approximation ratio can be achieved by an online algorithm. Identifying a class of supermartingales that admits a meaningful approximation ratio is an interesting direction for future research.

\subsection{Further Related Work}

To the best of our knowledge, ours is the first work to address the repeated buying and selling problem from a prophet inequality lens. In particular, the aspect that we allow the agent to buy and sell multiple times seems new.

On a technical level our work is related to the vast literature on prophet inequalities.
This includes work on the classic single-item prophet inequality \citep{KrengelS77,KrengelS78,SamuelC84},  the i.i.d.~case \citep{AbolhassaniEEHK17,CorreaFHOV21,KleinbergK18,Singla18,LiuLPSS21}, and the prophet secretary problem \citep{EsfandiariHLM17,EhsaniHKS18,CorreaSZ21}. Our unknown distributions result is related to a recent stream of work that looks at prophet inequalities with samples \citep{CorreaDFS19,RubinsteinWW20,CorreaCES20,CaramanisEtAl22}. Ideas from the prophet inequalities literature have also been used to construct simple near-optimal mechanisms for two-sided markets \citep{Colini-Baldeschi16,Colini-Baldeschi17,BraunK21,DuttingFLLR21}.

Relevant pointers into the literature on two-sided markets include the aforementioned seminal work by \citet{MyersonS83} and their celebrated impossibility result for the bilateral trade problem, as well as  \citep{McAfee08,Colini-Baldeschi17a,BrustleCWZ17,DengSW21} who show constant-factor approximation guarantees for maximizing 
the ``gains from trade''.

Related stopping problems have been studied in the finance/applied probability literature. First, \citet{GraversenEtAl06,DuToitP07} study the problem of stopping a Brownian motion, with the benchmark of minimizing the expected squared difference between the value at which the online algorithm stops and the maximum value in the sequence. Second, \citet{DuToitP09}, for the same stochastic process, study the problem of minimizing the expected ratio between the maximum value in the sequence and the value at which the online algorithm stops. In addition to studying different benchmarks, these works also do not quantify the worst-case performance guarantee.

There is also a rich literature on online portfolio selection problems, including work that takes a worst-case competitive analysis approach  \cite[Chapter 14]{BorodinEY98}. We refer the reader to the survey of \citet{LiHoi14} for details. 

\section{Model}
In our basic model we trade a single unit of an item. We are given $n$ distributions $F_1,\ldots, F_n$ and, in random order, in each of $n$ periods we are offered a price for the item, drawn from one of the distributions. More precisely, we have independent prices $X_i\sim F_i$ for $1\leq i\leq n$, and an independent uniformly random permutation $\sigma: \{1,\ldots,n\}\rightarrow \{1,\ldots,n\}$; and on period $i$ we get to observe the price $X_{\sigma(i)}$. In each period we have two possible states: We either have the item, or we do not. In each period $i$, immediately after observing the price $X_{\sigma(i)}$ we must make a decision. If we have the item, we can sell it at price $X_{\sigma(i)}$, or pass. If we do not have the item, we can buy it at price $X_{\sigma(i)}$, or pass. In period $1$ we do not have the item, and the state of period $i+1$ is determined by the decision of period $i$. 

Our objective is to design a decision rule that maximizes the expectation of the profit we make, where the profit is the sum of the prices on periods in which we sell minus the sum of prices on periods in which we buy. For an algorithm $\ALG$, slightly abusing notation, we also denote its (random) profit by $\ALG$ and its expected profit by $\E(\ALG)$.

We denote by $\OPT$ the optimal profit in hindsight. This is, $\OPT$ is the (random) maximum profit given the realization of the prices and the permutation $\sigma$, over all possible sequences of buy/sell decisions. We compare ourselves with $\E(\OPT)$, where the expectation is over the realizations of the prices and the permutation $\sigma$. We say a decision rule or algorithm $\ALG$ is an $\alpha$-approximation for a family of instances if
\[
\E(\ALG)\geq \alpha \cdot \E(\OPT)
\]
for all instances in that family. 

The following simple observation about $\OPT$, which can be proved by doing local changes in the decisions of $\OPT$, will be important in our analysis. It states that $\OPT$ buys in local minima and sells in local optima. See \cref{fig:opt} for an illustration.

\begin{observation}
\label{obs:opt_characterization}
To handle corner cases in a unified manner, define $X_{\sigma(0)} := \infty$ and $X_{\sigma(n+1)} := -\infty$. Then 
$\OPT$ buys in period $1 \leq i < n$ if and only if the price $X_{\sigma(i)}$ is a local minimum, i.e., $X_{\sigma(i)}<X_{\sigma(i-1)}$ and $X_{\sigma(i)}\leq X_{\sigma(i+1)}$. Analogously, $\OPT$ sells in period $1 < i \leq n$ if and only if the price $X_{\sigma(i)}$ is a local maximum, i.e., if 
$X_{\sigma(i)}\geq X_{\sigma(i-1)}$ and $X_{\sigma(i)}> X_{\sigma(i+1)}$.\footnote{The strict inequalities break ties in the border case where there is a set of consecutive periods with equal prices. In such a set $\OPT$ buys or sells at most once, and furthermore, it can do the operation in any period and the result is the same. With this particular choice, $\OPT$ will only buy in the first period of such a set, and sell only in the last one.}
\end{observation}

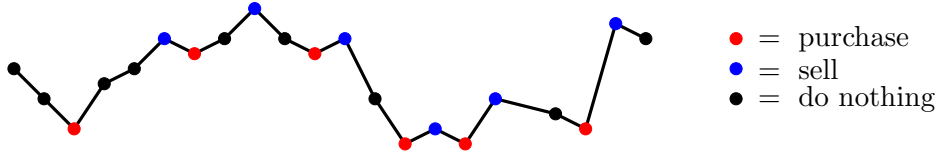
\begin{figure}
    \def\z{0.5}
    \centering
    \begin{tikzpicture}[scale=.8]
    \draw[very thick] 
    (0,\z*0) -- 
    (0.5,\z*-1) -- 
    (1.0,\z*-2) -- 
    (1.5,\z*-0.5) -- 
    (2.0,\z*0) -- 
    (2.5,\z*1) -- 
    (3.0,\z*0.5) -- 
    (3.5,\z*1.0) -- 
    (4.0,\z*2.0) -- 
    (4.5,\z*1) -- 
    (5.0,\z*0.5) -- 
    (5.5,\z*1.0) -- 
    (6.0,\z*-1.0) -- 
    (6.5,\z*-2.5) -- 
    (7,\z*-2) -- 
    (7.5,\z*-2.5) -- 
    (8,\z*-1) -- 
    (9,\z*-1.5) --
    (9.5,\z*-2) -- 
    (10,\z*1.5) -- 
    (10.5,\z*1);
    \draw[fill, color = black](0,\z*0) circle (0.1);
    \draw[fill, color = black](0.5,\z*-1) circle (0.1);
    \draw[fill, color = red](1.0,\z*-2) circle (0.1);
    \draw[fill, color = black](1.5,\z*-0.5) circle (0.1);
    \draw[fill, color = black](2.0,\z*0) circle (0.1);
    \draw[fill, color = blue](2.5,\z*1) circle (0.1);
    \draw[fill, color = red](3.0,\z*0.5) circle (0.1);
    \draw[fill, color = black](3.5,\z*1.0) circle (0.1);
    \draw[fill, color = blue](4.0,\z*2.0) circle (0.1);
    \draw[fill, color = black](4.5,\z*1) circle (0.1);
    \draw[fill, color = red](5.0,\z*0.5) circle (0.1);
    \draw[fill, color = blue](5.5,\z*1.0) circle (0.1);
    \draw[fill, color = black](6.0,\z*-1.0) circle (0.1);
    \draw[fill, color = red](6.5,\z*-2.5) circle (0.1);
    \draw[fill, color = blue](7,\z*-2) circle (0.1);
    \draw[fill, color = red](7.5,\z*-2.5) circle (0.1);
    \draw[fill, color = blue](8,\z*-1) circle (0.1);
    \draw[fill, color = black](9,\z*-1.5) circle (0.1);
    \draw[fill, color = red](9.5,\z*-2) circle (0.1);
    \draw[fill, color = blue](10,\z*1.5) circle (0.1);
    \draw[fill, color = black](10.5,\z*1) circle (0.1);
    \draw[fill, color = red] (12,0.5) circle (0.1);
    \node[anchor=west] at (12.2,0.5) {=\;  purchase};
    \draw[fill, color = blue] (12,0) circle (0.1);
    \node[anchor=west] at (12.2,0) {=\; sell};
    \draw[fill, color = black] (12,-0.5) circle (0.1);
    \node[anchor=west] at (12.2,-0.5) {=\; do nothing};
    \end{tikzpicture}
    \caption{Example of decisions made by $\OPT$, the optimal buy/sell decisions in hindsight.}
    \label{fig:opt}
\end{figure}

 In our proofs we assume all distributions are absolutely continuous. In particular, this implies that for every $i$ there always exists $T_i\in \R$ (a median) such that $\PP(X_i<T_i)=\PP(X_i\geq T_i)=1/2$. This assumption is w.l.o.g., as we can add to each price an auxiliary independent perturbation $\varepsilon_i\sim\text{Uniform}[-\varepsilon,\varepsilon]$, for an $\varepsilon>0$ such that $n\cdot \varepsilon \ll\E(\OPT)$.

\section{I.I.D.~Prices}
\label{sec:iid}

In this section we consider the case where all distributions are equal, so the prices are i.i.d. For this setting we remove the dependency on the permutation $\sigma$, and simply denote by $X_i$ the price in period $i$.
We present a simple algorithm that achieves a $1/2$-approximation, and show this factor is tight.

Our algorithm is a single-threshold rule, i.e., there is a threshold $T\in \R$ such that our decision in each period $1\leq i< n$ is to sell whenever $X_i\geq T$, and to buy whenever $X_i<T$. In period $n$ the algorithm sells if we have the item and passes if we do not, regardless of the price.
We denote such an algorithm as $\ALG_T$.

\begin{theorem}
\label{thm:iid_half_approx}
If the prices are i.i.d.~and $T$ is the median of the distribution, then $\ALG_T$ is a $1/2$-approximation, i.e., $\E(\ALG_T)\geq \frac{1}{2}\E(\OPT)$.
\end{theorem}

To prove this theorem, we first obtain in \Cref{lem:iid_upper_bound_opt} an upper bound on $\E(\OPT)$ in terms of $T$, and then show in \Cref{lem:iid_lower_bound_alg} that the expected profit of the algorithm is in fact half of this upper bound when $T$ is the median of the distribution.

\begin{lemma}
\label{lem:iid_upper_bound_opt}
If the prices are i.i.d., then for any $T\in \R$, 
\[
\E(\OPT)=\frac{n-1}{2}\cdot \E(|X_1-X_2|)\leq (n-1)\cdot \E(|X_1-T|).
\]
\end{lemma}
\begin{proof}
Let us denote by $\OPT_i$ the gain of $\OPT$ in period $i$, i.e., if $\OPT$ buys in period $i$, then $\OPT_i=-X_i$, if $\OPT$ sells in period $i$, then $\OPT_i=X_i$ and $\OPT_i=0$ otherwise. From \Cref{obs:opt_characterization}, we have the following.
\begin{align*}
    \OPT_i = \begin{cases} 
    -X_1\cdot \mathds{1}_{X_1\leq X_2} &\text{ if } i=1
    \\
    X_i\cdot \mathds{1}_{X_i\geq X_{i-1}} -X_i \cdot \mathds{1}_{X_i\leq X_{i+1}}
    &\text{ if }2\leq i\leq n-1
    \\
    X_n\cdot \mathds{1}_{X_n\geq X_{n-1}} 
    &\text{ if } i=n.
    \end{cases}
\end{align*}
Now we calculate $\E(\OPT)$ by adding up all these terms and taking expectation.
\begin{align*}
    \E(\OPT)&= \sum_{i=1}^n \E(\OPT_i)
    =\sum_{i=2} ^{n} \E(X_i\cdot \mathds{1}_{X_i\geq X_{i-1}})
    -\sum_{i=1} ^{n-1} \E(X_i \cdot \mathds{1}_{X_i\leq X_{i+1}})\\
    &=(n-1)\cdot \E\Big((X_1-X_2)\cdot \mathds{1}_{X_1\geq X_2}\Big).
\end{align*}
The last line follows from the fact that prices are i.i.d. Finally, also from the fact that prices are i.i.d.~and from the triangle inequality, we conclude that for any $T\in \R$,
\[
\E(\OPT)=\frac{n-1}{2}\cdot \E(|X_1-X_2|) \leq (n-1)\cdot \E(|X_1-T|).\qedhere
\]
\end{proof}

\begin{lemma}
\label{lem:iid_lower_bound_alg}
If the prices are i.i.d. and $T$ is the median of the distribution, then
\[
\E(\ALG_T) = \frac{n-1}{2}\cdot \E(|X_1-T|).
\]
\end{lemma}
\begin{proof}
As in \Cref{lem:iid_upper_bound_opt}, we analyze the  gains of $\ALG_T$ in period $i$, which we denote by $\ALG_T(i)$. If $\ALG_T$ buys in period $i$, then $\ALG_T(i)=-X_i$, if it sells, then $\ALG_T(i)=X_i$, and $\ALG_T(i)=0$ otherwise. By the definition of $\ALG_T$, we have that
\[
\ALG_T(1)=-X_1\cdot \mathds{1}_{X_1<T}.
\]
Denote by $p_i$ the probability that we have the item in period $i$ (before making a decision). For $2\leq i\leq n-1$, since the price of period $i$ is independent of the prices in periods $1,\ldots,i-1$, we have the following.
\begin{align*}
\E(\ALG_T(i))&= \E(X_i\cdot \mathds{1}_{X_i\geq T})\cdot p_i - 
\E(X_i\cdot \mathds{1}_{X_i< T}) \cdot (1-p_i).
\end{align*}
In the last period, the algorithm only sells whenever we have the item. Thus,
\[
\E(\ALG_T(n))=\E(X_n)\cdot p_n=\Big(\E(X_n\cdot \mathds{1}_{X_n\geq T})+ 
\E(X_n\cdot \mathds{1}_{X_n< T})\Big) \cdot p_n.
\]
We show now that if $T$ is the median of the distribution, $p_i=1/2$ for all $i\geq 2$. In fact, we have the following formula, which is obtained by conditioning on the event that we have (or not) the item in period $i-1$.
\begin{align*}
    p_i&=\PP(X_{i-1}<T)\cdot p_{i-1} +
    \PP(X_{i-1}<T)\cdot (1-p_{i-1}) \\
    &= \PP(X_{i-1}<T) = 1/2.
\end{align*}
Putting all together, we have that
\begin{align*}
    \E(\ALG_T) ={}& \sum_{i=1}^n \E(\ALG_T(i))\\
    ={}& -\E(X_1\cdot \mathds{1}_{X_1<T})
    +\frac{1}{2}\cdot\sum_{i=2}^{n-1} \E(X_i\cdot \mathds{1}_{X_i\geq T})
    -\frac{1}{2}\cdot\sum_{i=2}^{n-1}
    \E(X_i\cdot \mathds{1}_{X_i< T})\\
    &+\frac{1}{2}\cdot
    \Big(\E(X_n\cdot \mathds{1}_{X_n\geq T})+ 
    \E(X_n\cdot \mathds{1}_{X_n< T})\Big)\\
    ={}&
    \frac{n-1}{2}\cdot
    \Big( \E(X_1\cdot \mathds{1}_{X_1\geq T}) - \E(X_1\cdot \mathds{1}_{X_1<T}) \Big).
\end{align*}
Here, the last line comes from the fact that the prices are i.i.d. Finally, notice that $\PP(X_1\geq T)= \PP(X_1<T)=1/2$, so $\E(T\cdot \mathds{1}_{X_1\geq T})=\E(T\cdot \mathds{1}_{X_1<T})$, and therefore, by adding and subtracting this quantity in the last line, we conclude that
\[
\E(\ALG_T)=\frac{n-1}{2}\cdot \E(|X_1-T|). \qedhere
\]
\end{proof}

\begin{proposition}\label{prop:iid_lb}
For any $n \ge 2$ and $\varepsilon>0$, there is an instance with $n$ i.i.d.~prices such that for the optimal algorithm \ALG,
$$\E(\ALG) \leq \left(\frac{1}{2}+\varepsilon\right) \cdot\E(\OPT).$$
\end{proposition}
The proof can be found in~\cref{appx:subsec:iid_lb}.

\section{Non-I.I.D.~Prices, Random Order}
\label{sec:rand-order}

In this section we consider the case where the prices $X_1, \ldots, X_n$ are independent draws from not necessarily identical distributions $F_1, \ldots, F_n$; and these prices are presented to us in random order. 
We show three results, a $1/16$ approximation by a threshold policy, an impossibility showing that no online policy can achieve a better than $1/3$ approximation, and an asymptotic $1/2-o(1)$ approximation by a threshold policy as $n \rightarrow \infty$. The asymptotic $1/2 - o(1)$ approximation is obtained by setting the threshold to the median of the mixture distribution. The  $1/16$ approximation requires a different threshold. We describe how we choose the threshold for the $1/16$ approximation in~\cref{sec:uniform}, and provide a hard instance for the median of the mixture distribution for small $n$ in the appendix.

\begin{theorem}\label{thm:random-order}
If the prices are presented in 
order $X_{\sigma(1)}, \ldots, X_{\sigma(n)}$, where $\sigma : \{1, \ldots, n\} \rightarrow \{1, \ldots, n\}$ is a uniform random permutation,
then there exist a threshold $T$ such that  
$$\E(\ALG_{T}) \ge  \frac{1}{16}\cdot\E(\OPT).$$
Moreover, the threshold $T$ can be computed in polynomial time.
\end{theorem}

\begin{proposition}\label{prop:rdm_order_lb}
For every $\varepsilon>0$, there is an instance such that if the prices arrive in uniform random order, then for the optimal algorithm $\ALG$,
$$ \E(\ALG) \leq \left(\frac{1}{3}+\varepsilon\right) \cdot \E(\OPT).$$
\end{proposition}

\begin{theorem}\label{thm:random-order-v2}
If the prices are presented in 
order $X_{\sigma(1)}, \ldots, X_{\sigma(n)}$, where $\sigma : \{1, \ldots, n\} \rightarrow \{1, \ldots, n\}$ is a uniform random permutation,
then there exist a threshold $T$ and a constant $C$, independent of these distributions, such that  
$$\bigg(1 + \frac{C}{n} \bigg) \cdot 2 \ge  \frac{\E(\OPT)}{ \E(\ALG_{T})}.$$
\end{theorem}

We provide the proofs of \cref{thm:random-order} and \cref{thm:random-order-v2} in \cref{sec:reduction}, \cref{sec:uniform}, and \cref{sec:asymptotic}; the proof of \cref{prop:rdm_order_lb} in \cref{appx:rdm_order_lb}. 

\subsection{Reduction to Two Periods with Correlated Random Variables}\label{sec:reduction}

We start by reducing the problem of showing an approximation guarantee for the random order model with $n$ periods to a two period problem with correlated random variables.

\begin{lemma}\label{lemma:opt-est-ran-ord}
If the prices are presented in uniformly random order $\sigma$, then
$$\E(\OPT) = \frac{n - 1}{2} \cdot \E(|X_{\sigma(1)} - X_{\sigma(2)}|).$$
\end{lemma}

\begin{proof}
For a fixed order $\sigma$ and realizations $X_{\sigma(1)}, \ldots, X_{\sigma(n)}$, we can argue as in the beginning of the proof of \cref{lem:iid_upper_bound_opt} and use \cref{obs:opt_characterization} to obtain
\begin{align*}
    \OPT = \sum_{i=2}^{n} X_{\sigma(i)} \cdot \mathds{1}_{X_{\sigma(i)} \geq X_{\sigma(i-1)}} - \sum_{i=1}^{n-1} X_{\sigma(i)} \cdot \mathds{1}_{X_{\sigma(i)} \leq X_{\sigma(i+1)}}
    = \sum_{i=1}^{n-1} ( X_{\sigma(i+1)} - X_{\sigma(i)}) \cdot \mathds{1}_{X_{\sigma(i+1)} \geq X_{\sigma(i)}}.
\end{align*}

Then, taking expectation over the possible orders $\sigma$ and realizations $X_{\sigma(1)}, \ldots, X_{\sigma(n)}$,
\begin{align*}
    \E(\OPT) 
&= \sum_{i = 1}^{n - 1}\E( (X_{\sigma(i + 1)} - X_{\sigma(i)}) \cdot \mathds{1}_{X_{\sigma(i + 1)} \geq X_{\sigma(i)}})\\
&= (n - 1) \cdot \E((X_{\sigma(1)} - X_{\sigma(2)}) \cdot \mathds{1}_{X_{\sigma(1)} \geq X_{\sigma(2)}}) 
= \frac{n -1}{2}\cdot \E(|X_{\sigma(1)} - X_{\sigma(2)}|),
\end{align*}
where the first equality holds by linearity of expectation and the second and third equality hold because $\sigma$ is a uniformly random order.
\end{proof}

\begin{lemma}\label{lemma:alg-est-ran-ord}
If the prices are presented in uniformly random order $\sigma$, then for any threshold $T \in \R$ we have that
$$\E(\ALG_{T}) = \frac{n - 1}{2} \cdot \E(|X_{\sigma(1)} - X_{\sigma(2)}| \cdot \mathds{1}_{T \in [\min(X_{\sigma(1)}, X_{\sigma(2)}), \max(X_{\sigma(1)}, X_{\sigma(2)})]}).$$
\end{lemma}
\begin{proof}
Let $\ALG_T(i)$ denote the gain (or the loss) of $\ALG_T$ in period $i$. By the definition of $\ALG_T$, it can only buy in period $i = 1$. Thus,
$\ALG_T(1) = - X_{\sigma(1)} \cdot \mathds{1}_{X_{\sigma(1)} < T}.$

Consider now any period $2 \le i \le n-1$. In such period, the algorithm buys whenever $X_{\sigma(i)} < T$ and it does not have the item. The latter event is equivalent to the event $X_{\sigma(i - 1)} \ge T$. Indeed, assuming that $X_{\sigma(i - 1)} \ge T$ either $\ALG_{T}$ in period $i-1$ does not have the item, but in this case it will not buy it in this period; or it does have the item, but then it will certainly sell the item in period $i - 1$. In both cases, the algorithm does not have the item in period $i$. On the other hand, if $X_{\sigma(i-1) < T}$ then either the algorithm already has the item in period $i-1$ or it will buy it in that period, so in this case it will certainly hold the item in period $i$.

Selling an item is similar. In any period $2 \le i \le n-1$ the algorithm sells the item iff $X_{\sigma(i)} \ge T$ and it has the item. Analogously, the algorithm has the item in period $i$ iff $X_{\sigma(i - 1)} < T$.
Therefore, we can express the gains of $\ALG_T$ in period $ 2\leq i \leq n-1$ by
$$\ALG_T(i) = -X_{\sigma(i)} \cdot \mathds{1}_{X_{\sigma(i)} < T} \cdot \mathds{1}_{X_{\sigma(i - 1)} \ge T} + X_{\sigma(i)} \cdot \mathds{1}_{X_{\sigma(i)} \ge T} \cdot \mathds{1}_{X_{\sigma(i - 1)} < T}.$$
In the last period, the algorithm never buys and always sells if it has the item. Thus, from the same reasons as above we have that
$\ALG_T(n) = X_{\sigma(n)} \cdot \mathds{1}_{X_{\sigma(n - 1)} < T}.$
Putting all together we obtain
\begin{align*}
    \E(\ALG_T) ={}& \sum_{i=1}^n \E(\ALG_T(i))\\
    ={}& -\E(X_{\sigma(1)}\cdot \mathds{1}_{X_{\sigma(1)} < T})
    +\sum_{i=2}^{n-1} \E(X_{\sigma(i)} \cdot \mathds{1}_{X_{\sigma(i)} \ge T} \cdot \mathds{1}_{X_{\sigma(i - 1)} < T})
    \\[-6pt]
    &\hspace*{48pt}-\sum_{i=2}^{n-1}
    \E(X_{\sigma(i)} \cdot \mathds{1}_{X_{\sigma(i)} < T} \cdot \mathds{1}_{X_{\sigma(i - 1)} \ge T}) +\E( X_{\sigma(n)} \cdot \mathds{1}_{X_{\sigma(n - 1)} < T})\\
    \stackrel{(\ast)}{=}{}&
    (n-2)\cdot
    \E(X_{\sigma(1)} \cdot \mathds{1}_{X_{\sigma(1)} \ge T} \cdot \mathds{1}_{X_{\sigma(2)} < T} -X_{\sigma(2)} \cdot \mathds{1}_{X_{\sigma(2)} < T} \cdot \mathds{1}_{X_{\sigma(1)} \ge T} ) \\
    &\hspace*{48pt}+ \E(X_{\sigma(1)} \cdot \mathds{1}_{X_{\sigma(2)} < T})  - \E(X_{\sigma(1)}\cdot \mathds{1}_{X_{\sigma(1)} < T}),
\end{align*}
where $(\ast)$ follows from  linearity of expectation and the fact that $\sigma$ is a uniformly random permutation. 
By manipulating the last two terms of the above sum, we obtain
\begin{align*}
&\E(X_{\sigma(1)} \cdot \mathds{1}_{X_{\sigma(2)} < T})  - \E(X_{\sigma(1)}\cdot \mathds{1}_{X_{\sigma(1)} < T}) \\
&\qquad\qquad= \E(X_{\sigma(1)} \cdot \mathds{1}_{X_{\sigma(2)} < T} \cdot \mathds{1}_{X_{\sigma(1)} < T} + X_{\sigma(1)} \cdot \mathds{1}_{X_{\sigma(2)} < T} \cdot \mathds{1}_{X_{\sigma(1)} \ge T})  \\
&\hspace*{96pt}- \E(X_{\sigma(1)}\cdot \mathds{1}_{X_{\sigma(1)} < T} \cdot \mathds{1}_{X_{\sigma(2)} \ge T} + X_{\sigma(1)}\cdot \mathds{1}_{X_{\sigma(1)} < T}\cdot \mathds{1}_{X_{\sigma(2)} < T})\\
&\qquad\qquad= \E(X_{\sigma(1)} \cdot \mathds{1}_{X_{\sigma(2)} < T} \cdot \mathds{1}_{X_{\sigma(1)} \ge T}) - \E(X_{\sigma(1)}\cdot \mathds{1}_{X_{\sigma(1)} < T} \mathds{1}_{X_{\sigma(2)} \ge T}) \\
&\qquad\qquad=\E(X_{\sigma(1)} \cdot \mathds{1}_{X_{\sigma(1)} \ge T} \cdot \mathds{1}_{X_{\sigma(2)} < T} -X_{\sigma(2)} \cdot \mathds{1}_{X_{\sigma(2)} < T} \cdot \mathds{1}_{X_{\sigma(1)} \ge T}).
\end{align*}
Substituting this back into the formula for $\E(\ALG_T)$ yields
\begin{align*}
    \E(\ALG_T) ={}& (n - 1) \cdot \E(X_{\sigma(1)} \cdot \mathds{1}_{X_{\sigma(1)} \ge T} \cdot \mathds{1}_{X_{\sigma(2)} < T} -X_{\sigma(2)} \cdot \mathds{1}_{X_{\sigma(2)} < T} \cdot \mathds{1}_{X_{\sigma(1)} \ge T})\\
    ={}& (n - 1) \cdot \E((X_{\sigma(1)} - X_{\sigma(2)}) \cdot \mathds{1}_{X_{\sigma(1)} \ge T} \cdot \mathds{1}_{X_{\sigma(2)} < T} )\\
    ={}& \frac{n - 1}{2} \cdot \E(|X_{\sigma(1)} - X_{\sigma(2)}| \cdot \mathds{1}_{T \in [\min(X_{\sigma(1)}, X_{\sigma(2)}), \max(X_{\sigma(1)}, X_{\sigma(2)})]}),
\end{align*}
where the last lines follows from the fact that $\sigma$ is a uniformly random permutation.
\end{proof}

\subsection{Proof of \cref{thm:random-order}}\label{sec:uniform}

\cref{lemma:opt-est-ran-ord} and \cref{lemma:alg-est-ran-ord} imply that, 
in order to show \cref{thm:random-order}, it suffices to show that there exists a threshold $T \in \R$ such that
\begin{equation}\label[ineq]{eq:5}
\E(|X_{\sigma(1)} - X_{\sigma(2)}|\cdot \mathds{1}_{T \in [\min(X_{\sigma(1)}, X_{\sigma(2)}), \max(X_{\sigma(1)}, X_{\sigma(2)})]}) \ge \frac{1}{16} \cdot \E(|X_{\sigma(1)} - X_{\sigma(2)}|). 
\end{equation}

The main difficulty in showing \cref{eq:5} is in the fact that the random variables $X_{\sigma(1)}$ and $X_{\sigma(2)}$ are \emph{not} independent. 
Indeed, if they were independent then the inequality would be implied by the following key lemma. 
\begin{lemma}\label{lemma:two-medians}
Let $X_{1}, X_{2}$ be two independent prices with distributions $F_{1}, F_{2}$. Then, there exists a threshold $T \in \R$ such that
\begin{equation}\label[ineq]{eq:1}
\E(|X_{1} - X_{2}| \cdot \mathds{1}_{T \in [\min(X_{1}, X_{2}), \max(X_{1}, X_{2})]}) \ge \frac{1}{4}\cdot\E(|X_{1} - X_{2}|)    
\end{equation}
\end{lemma}

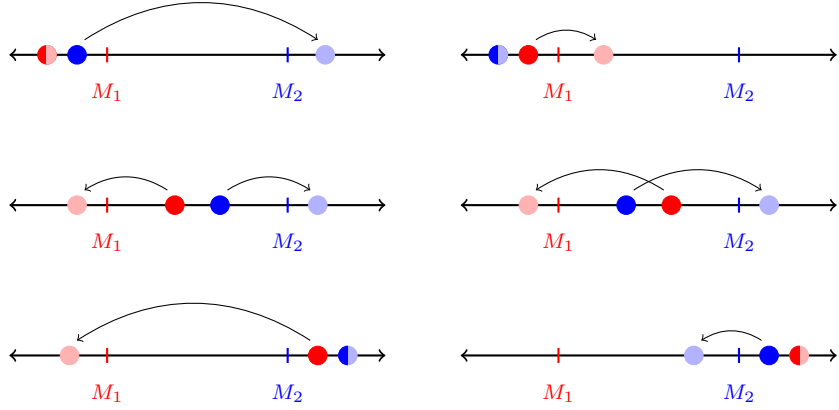
\begin{figure}
    \centering
    \begin{tikzpicture}
        \begin{scope}[shift={(0,0)}]
            \draw[<->,thick] (0,0) -- (5,0);
            \draw[thick,red] (1.3,-0.1) -- (1.3,0.1);
            \node at (1.3,-0.5) {\scriptsize \textcolor{red}{$M_1$}};
            \draw[thick,blue] (3.7,-0.1) -- (3.7,0.1);
            \node at (3.7,-0.5) {\scriptsize \textcolor{blue}{$M_2$}};
            \filldraw[red] (0.5,0) circle (3.5pt);
            \begin{scope}
                \clip(0.5,-0.5) rectangle (2,0.5);
                \filldraw[red!30] (0.5,0) circle (3.5pt);
            \end{scope}
            \filldraw[blue] (0.9,0) circle (3.5pt);
            \draw[->] (1,0.2) to [bend left=33] (4.1,0.2);
            \filldraw[blue!30] (4.2,0) circle (3.5pt);
        \end{scope}    
        \begin{scope}[shift={(6,0)}]
            \draw[<->,thick] (0,0) -- (5,0);
            \draw[thick,red] (1.3,-0.1) -- (1.3,0.1);
            \node at (1.3,-0.5) {\scriptsize \textcolor{red}{$M_1$}};
            \draw[thick,blue] (3.7,-0.1) -- (3.7,0.1);
            \node at (3.7,-0.5) {\scriptsize \textcolor{blue}{$M_2$}};
            \filldraw[blue] (0.5,0) circle (3.5pt);
            \begin{scope}
                \clip(0.5,-0.5) rectangle (2,0.5);
                \filldraw[blue!30] (0.5,0) circle (3.5pt);
            \end{scope}
            \filldraw[red] (0.9,0) circle (3.5pt);
            \draw[->] (1,0.2) to [bend left=33] (1.8,0.2);
            \filldraw[red!30] (1.9,0) circle (3.5pt);
        \end{scope}
        \begin{scope}[shift={(0,-2)}]
            \draw[<->,thick] (0,0) -- (5,0);
            \draw[thick,red] (1.3,-0.1) -- (1.3,0.1);
            \node at (1.3,-0.5) {\scriptsize \textcolor{red}{$M_1$}};
            \draw[thick,blue] (3.7,-0.1) -- (3.7,0.1);
            \node at (3.7,-0.5) {\scriptsize \textcolor{blue}{$M_2$}};
            \filldraw[red] (2.2,0) circle (3.5pt);
            \filldraw[blue] (2.8,0) circle (3.5pt);
            \draw[->] (2.1,0.2) to [bend right=33] (1.0,0.2);
            \filldraw[red!30] (0.9,0) circle (3.5pt);
            \draw[->] (2.9,0.2) to [bend left=33] (4,0.2);
            \filldraw[blue!30] (4.1,0) circle (3.5pt);            
        \end{scope}    
        \begin{scope}[shift={(6,-2)}]
            \draw[<->,thick] (0,0) -- (5,0);
            \draw[thick,red] (1.3,-0.1) -- (1.3,0.1);
            \node at (1.3,-0.5) {\scriptsize \textcolor{red}{$M_1$}};
            \draw[thick,blue] (3.7,-0.1) -- (3.7,0.1);
            \node at (3.7,-0.5) {\scriptsize \textcolor{blue}{$M_2$}};
            \filldraw[blue] (2.2,0) circle (3.5pt);
            \filldraw[red] (2.8,0) circle (3.5pt);
            \draw[->] (2.7,0.2) to [bend right=33] (1.0,0.2);
            \filldraw[red!30] (0.9,0) circle (3.5pt);
            \draw[->] (2.3,0.2) to [bend left=33] (4,0.2);
            \filldraw[blue!30] (4.1,0) circle (3.5pt);         
        \end{scope}       
        \begin{scope}[shift={(5,-4)},xscale=-1]
            \draw[<->,thick] (0,0) -- (5,0);
            \draw[thick,blue] (1.3,-0.1) -- (1.3,0.1);
            \node at (1.3,-0.5) {\scriptsize \textcolor{blue}{$M_2$}};
            \draw[thick,red] (3.7,-0.1) -- (3.7,0.1);
            \node at (3.7,-0.5) {\scriptsize \textcolor{red}{$M_1$}};
            \filldraw[blue!30] (0.5,0) circle (3.5pt);
            \begin{scope}
                \clip(0.5,-0.5) rectangle (2,0.5);
                \filldraw[blue] (0.5,0) circle (3.5pt);
            \end{scope}
            \filldraw[red] (0.9,0) circle (3.5pt);
            \draw[->] (1,0.2) to [bend left=33] (4.1,0.2);
            \filldraw[red!30] (4.2,0) circle (3.5pt);
        \end{scope}    
        \begin{scope}[shift={(11,-4)},xscale=-1]
            \draw[<->,thick] (0,0) -- (5,0);    
            \draw[thick,blue] (1.3,-0.1) -- (1.3,0.1);
            \node at (1.3,-0.5) {\scriptsize \textcolor{blue}{$M_2$}};
            \draw[thick,red] (3.7,-0.1) -- (3.7,0.1);
            \node at (3.7,-0.5) {\scriptsize \textcolor{red}{$M_1$}};
            \filldraw[red!30] (0.5,0) circle (3.5pt);
            \begin{scope}
                \clip(0.5,-0.5) rectangle (2,0.5);
                \filldraw[red] (0.5,0) circle (3.5pt);
            \end{scope}
            \filldraw[blue] (0.9,0) circle (3.5pt);
            \draw[->] (1,0.2) to [bend left=33] (1.8,0.2);
            \filldraw[blue!30] (1.9,0) circle (3.5pt);
        \end{scope}               
    \end{tikzpicture}
    \caption{The different cases distinguished in the proof of Lemma~\ref{lemma:two-medians}. Realizations of $X_1$ and $X_2$ are shown in red and blue, respectively. The realizations that are charged to are depicted in lighter colors. While the specific values depend on the actual distributions, what is unaffected is the order of the realizations with respect to the median of the corresponding distribution as well as among each other.}
    \label{fig:tech-proof}
\end{figure}
In the appendix, we show that \cref{eq:1} can be fulfilled by setting $T$ to at least one of $M_1$ and $M_2$, the medians of $F_1$ and $F_2$, respectively. Note that this is implied by
\begin{align*}
& 2\left( \E(|X_{1} - X_{2}| \cdot \mathds{1}_{M_{1} \in [\min(X_{1}, X_{2}), \max(X_{1}, X_{2})]}) + \E(|X_{1} - X_{2}| \cdot \mathds{1}_{M_{2} \in [\min(X_{1}, X_{2}), \max(X_{1}, X_{2})]}) \right) \notag\\ &\hspace*{48pt}\ge \E(|X_{1} - X_{2}|).
\end{align*}
A way to view the proof is that we charge any elementary event, say, $X_1=x_1$ and $X_2=x_2$, to another elementary event $X_1=x_1',X_2=x_2'$ such that
\begin{compactenum}
    \item[(i)] $|x_1-x_2|\leq |x_1'-x_2'|$,
    \item[(ii)] $M_1\in[\min(x_1',x_2'),\max(x_1',x_2')]$ or $M_2\in[\min(x_1',x_2'),\max(x_1',x_2')]$,
    \item[(iii)] no elementary event is charged to more than two times.
\end{compactenum}
First note that we can charge each event $X_1=x_1,X_2=x_2$ satisfying $M_1\in[\min(x_1,x_2),\max(x_1,x_2)]$ or $M_2\in[\min(x_1,x_2),\max(x_1,x_2)]$ to itself. For the other events, we distinguish different cases visualized in~\cref{fig:tech-proof} and charge in such a way that each realization of the former type is only charged to one additional time. Due to the presence of long calculations in the formal proof, we refer reader to \cref{appx:subsec:two-medians} for details. In the appendix, we also show that the analysis is tight for this way of setting the threshold.

Using the lemma, we now proceed to the proof of the theorem.

\begin{proof}[Proof of \cref{thm:random-order}]

The goal of this proof is constructing two \textit{independent} random variables such that if they are put to \cref{eq:5} they, with a loss of only a constant factor, estimate the right-hand side from above and the left-hand side from below.

Let set $H_{1} \subseteq \{1, \ldots, n\}$ be a set chosen uniformly at random from all subset of size $\ceil{\frac{n}{2}}$. Let $a, b$ be random elements chosen uniformly from $H_{1}$ and $\{1, \ldots, n\} \setminus H_{1}$. For any two $i \neq j, i,j \in \{1, \ldots, n\}$ we have that
$$\PP(a = i, b = j) =   \PP(\sigma(1) = i, \sigma(2) = j) = \frac{1}{n(n - 1)},$$
which gives us that
$$\E(|X_{a} - X_{b}|) = \E(|X_{\sigma(1)} - X_{\sigma(2)}|).$$
Since choice of each subset of size $\ceil{\frac{n}{2}}$ of set $\{1, \ldots, n \}$ is equally likely and happens with probability $1/{n \choose \ceil{n / 2}}$, there must exist a set $S$ such that
\begin{equation}\label[ineq]{eq:3}
\E(|X_{a'} - X_{b'}|) \ge \E(|X_{\sigma(1)} - X_{\sigma(2)}|),    
\end{equation}
where $a'$ is a uniformly random element from $S$ while $b'$ is a uniformly random element from $\{1, \ldots, n \} \setminus S$. Note here that $X_{a'}$ and $X_{b'}$ are independent random variables since the sets of random variables $\{ X_{i} | i \in S\}$ and $\{ X_{i} | i \in \{1, \ldots, n\} \setminus S\}$ are  pair-wise independent. Therefore, we can apply~\cref{lemma:two-medians} and get a threshold $T \in \R$ such that
$$\E(|X_{a'} - X_{b'}| \cdot \mathds{1}_{T \in [\min(X_{a'}, X_{b'}), \max(X_{a'}, X_{b'})]}) \ge \frac{1}{4}\E(|X_{a'} - X_{b'}|).$$
This together with \cref{eq:3} yields
\begin{equation}\label[ineq]{eq:4}
\E(|X_{a'} - X_{b'}| \cdot \mathds{1}_{T \in [\min(X_{a'}, X_{b'}), \max(X_{a'}, X_{b'})]}) \ge \frac{1}{4}\E(|X_{\sigma(1)} - X_{\sigma(2)}|).
\end{equation}
To the end, observe that 
\begin{align*}
&4 \cdot \E(|X_{\sigma(1)} - X_{\sigma(2)}|\cdot \mathds{1}_{T \in [\min(X_{\sigma(1)}, X_{\sigma(2)}), \max(X_{\sigma(1)}, X_{\sigma(2)})]})\\
&\qquad= \sum_{i,j \in \{1, \ldots, n\}, i \neq j} \frac{4}{n(n-1)} \E(|X_i - X_j|\cdot \mathds{1}_{T \in [\min(X_{i}, X_{j}), \max(X_{i}, X_{j})]}) \\
&\qquad\ge \sum_{i \in S, j \in \{1, \ldots, n\} \setminus S} \frac{1}{\ceil{n / 2}} \cdot \frac{1}{n - \ceil{n / 2}}\E(|X_i - X_j|\cdot \mathds{1}_{T \in [\min(X_{i}, X_{j}), \max(X_{i}, X_{j})]}) \\
&\qquad= \E(|X_{a'} - X_{b'}| \cdot \mathds{1}_{T \in [\min(X_{a'}, X_{b'}), \max(X_{a'}, X_{b'})]}),
\end{align*}
where the inequality above follows from the following reasons. First, all random variables on both sides of the inequality are non-negative. Second, on the left-hand side the sum iterates over all possible distinct pairs $i, j$ while on the right-hand side the sum iterates only over the pairs for which $i$ belongs to $S$ and $j$ belongs to the complement of $S$. Third, a simple case analysis assures that $\frac{4}{n(n - 1)} = \frac{2}{n}\cdot\frac{2}{n - 1} \ge \frac{1}{\ceil{n / 2}} \cdot \frac{1}{n - \ceil{n / 2}}$ for all integer $n$.
This together with \cref{eq:4} proves the theorem.
\end{proof}

\subsection{Proof of \cref{thm:random-order-v2}}\label{sec:asymptotic}

Let us now turn to proving \cref{thm:random-order-v2}. We will be analyzing the fraction
$$\frac{\E(|X_{1} - X_{2}|)}{\E(|X_{1} - X_{2}| \cdot \mathds{1}_{T \in [\min(X_{1}, X_{2}), \max(X_{1}, X_{2})]})},$$
where $T$ is the median of the mixture random variable $\frac{1}{n}X_{1} + \ldots + \frac{1}{n}X_{n}$. We will relate the nominator and denominator of the above fraction to two independent random variables of the \textit{same} distribution which makes the key difference compared to the previous approach.

Consider a random vector $(X_{1}, \ldots, X_{n})$ and its independent copy with the same distribution $(X'_{1}, \ldots, X'_{n})$. Let $Y, Y'$ be two independent random numbers from $1$ to $n$. Then, the variables $X_{Y}$ and $X'_{Y'}$ are also independent and have the same distribution.  

We define
\begin{align*}
S_{1} := \E(|X_{\sigma(1)} - X_{\sigma(2)}|), \text{ and }
S_2 := \E(|X_{\sigma(1)} - X_{\sigma(2)}| \cdot \mathds{1}_{T \in [\min(X_{\sigma(1)}, X_{\sigma(2)}), \max(X_{\sigma(1)}, X_{\sigma(2)})]})
\end{align*}
as well as
\begin{align*}
S'_1 := \E(|X_{Y} - X'_{Y'}|), \text{ and } S'_2 := \E(|X_{Y} - X'_{Y'}| \cdot \mathds{1}_{T \in [\min(X_{Y}, X'_{Y'}), \max(X_{Y}, X'_{Y'})]}).
\end{align*}
Since $T$ is the median of $X_{Y}$ and $X'_{Y'}$, thus by \cref{lem:iid_upper_bound_opt} we have that $2 \ge S'_{1}/S'_{2}$. The following lemma relates $S'_{2}$ to $S_{2}$.
\begin{lemma}\label{lem:dem}
If $T$ is the median of distribution $X_{Y}$, then there exists a constant $C > 0$, independent from $X_{Y}$, such that the following inequality holds:
$$\bigg(1 + \frac{C}{n}\bigg)\cdot S_{2} \ge \frac{n}{n - 1}\cdot S'_{2} .$$
\end{lemma}

\begin{proof}
Throughout this proof, for any real $A,B$, we will use $[A, B]$ to denote a continuous interval $[\min(A, B), \max(A, B)]$. We have that 
\begin{align*}
\frac{n}{n - 1}\cdot S'_{2} = \frac{n}{n - 1}\cdot\E(|X_{Y} - X'_{Y'}| \cdot \mathds{1}_{T \in [X_{Y}, X'_{Y'}]}) = \frac{1}{(n-1)n}\cdot \sum_{i, i' \in [n]}  \frac{1}{n^2} \cdot \E(|X_{i} - X'_{i'}| \cdot \mathds{1}_{T \in [X_{i}, X'_{i'}]}).
\end{align*}
On the other hand it holds that
\begin{align*}
\bigg(1 + \frac{C}{n}\bigg)\cdot S_{2} =\; & \bigg(1 + \frac{C}{n}\bigg)\cdot \E(|X_{\sigma(1)} - X_{\sigma(2)}| \cdot \mathds{1}_{T \in [X_{\sigma(1)}, X_{\sigma(2)}]})\\
=\; &\bigg(1 + \frac{C}{n}\bigg)\cdot \sum_{i \neq j \in [n]} \frac{1}{(n - 1)n} \cdot \E(|X_{i} - X_{j}| \cdot \mathds{1}_{T \in [X_{i}, X_{j}]}).
\end{align*}
Since $X_{i}$ and $X'_{j}$ are independent if $i \neq j$ it makes sense to subtract
$\sum_{i \neq j \in [n]} \frac{1}{(n - 1)n} \cdot \E(|X_{i} - X_{j}| \cdot \mathds{1}_{T \in [X_{i}, X_{j}]})$
from both sides of the original inequality, which reduces our task to proving
\begin{align*}
\frac{C}{(n-1)n^2} \cdot\sum_{i \neq j, i,j \in [n]} \E(|X_{i} - X'_{j}| \cdot \mathds{1}_{T \in [X_{i}, X'_{j}]} \ge \frac{1}{n(n-1)}\cdot \sum_{i \in [n]} \E(|X_{i} - X'_{i}| \cdot \mathds{1}_{T \in [X_{i}, X'_{i}]}).
\end{align*}
Let us now fix $i \in [n]$. We will show that
\begin{align}\label[ineq]{ineq:7}
\frac{C}{n} \cdot\sum_{j \in [n], j \neq i} \E(|X_{i} - X'_{j}| \cdot \mathds{1}_{T \in [X_{i}, X'_{j}]})  \ge \E(|X_{i} - X'_{i}| \cdot \mathds{1}_{T \in [X_{i}, X'_{i}]}),  
\end{align}
which, if summed over all choices of $i \in [n]$, will prove the lemma.
To do so, let $A$ be the set of the indices $k$ for which $\PP(X'_{k} > T) \ge 1/4$. Observe that $|A| \ge n/4$. If not, then we have $\PP({X'_{Y}} > T) < \PP(Y \in A) + \frac{1}{4}\cdot\PP(Y \notin A) < \frac{1}{4} + \frac{1}{4} = \frac{1}{2},$ 
which contradicts with the choice of $T$. An analogous argument shows that set
$B := \bigg\{k \mid \PP(X'_{k} \le T) \ge \frac{1}{4} \bigg\}$ 
has size at least $n/4$. 
Consider now the sum $\sum_{j \in [n], j \neq {i}} \E(|X_{i} - X'_{j}| \cdot \mathds{1}_{T \in [X_{i}, X'_{j}]}$. 
We have that
\begin{align}
&2\sum_{j \in [n], j \neq i} \E(|X_{i} - X'_{j}| \cdot \mathds{1}_{T \in [X_{i}, X'_{j}]}\nonumber\\
\ge& \sum_{a \in A - i} \E(|X_{i} - X'_{a}| \cdot \mathds{1}_{X_{i} \le T} \cdot \mathds{1}_{X'_{a} > T]}) + \sum_{b \in B - i} \E(|X_{i} - X'_{b}| \cdot \mathds{1}_{X_{i} > T} \cdot \mathds{1}_{X'_{b} \le T})\nonumber\\
=& \sum_{a \in A - i} \E((X'_{a} - X_{i}) \cdot \mathds{1}_{X_{i} \le T} \cdot \mathds{1}_{X'_{a} > T}) + \sum_{b \in B - i} \E((X_{i} - X'_{b}) \cdot \mathds{1}_{X_{i} > T} \cdot \mathds{1}_{X'_{b} \le T})\nonumber\\
\ge& \sum_{a \in A - i} \E((T - X_{i}) \cdot \mathds{1}_{X_{i} \le T} \cdot \mathds{1}_{X_{a} > T}) + \sum_{b \in B - i} \E((X_{i} - T) \cdot \mathds{1}_{X_{i} > T} \cdot \mathds{1}_{X_{b} \le T}),\label{eq:8.7}
\end{align}
where the last inequality follows from the fact that $X_{a} > T$ and $X_{b} \le T$. Since for every $a \in A - i$ variable $X'_{a}$ is independent from $X_{i}$ and, by the choice of $A$, we have $\PP(X_{a} > T) \ge 1/4$, it holds that
\begin{align*}
\sum_{a \in A - i} \E((T - X_{i}) \cdot \mathds{1}_{X_{i} \le T} \cdot \mathds{1}_{X'_{a} > T}) 
\ge \frac{1}{4} \cdot \bigg(\frac{n}{4} - 1 \bigg)\cdot \E((T - X_{i}) \cdot \mathds{1}_{X_{i} \le T}).
\end{align*}
The last inequality follows from the fact that $|A| \ge \frac{n}{4}$.
By a symmetric reasoning for set $B$, we see that
\begin{align*}
   \sum_{b \in B - i} \E((X_{i} - T) \cdot \mathds{1}_{X_{i} > T} \cdot \mathds{1}_{X'_{b} \le T}) \ge \frac{1}{4} \cdot \bigg(\frac{n}{4} - 1 \bigg)\cdot\E((X_{i} - T) \cdot \mathds{1}_{X_{i} > T}). 
\end{align*}
The two above inequalities combined with Inequality~\eqref{eq:8.7} gives us that
\begin{align}\label[ineq]{eq:6.10}
\sum_{j \in [n], j \neq i} 2\E(|X_{i} - X'_{j}| \cdot \mathds{1}_{T \in [X_{i}, X'_{j}]}) \ge \bigg(\frac{n}{16} - \frac{1}{4}\bigg)\cdot \big(\E((T - X_{i}) \cdot \mathds{1}_{X_{i} \le T}) + \E((X_{i} - T) \cdot \mathds{1}_{X_{i} > T}) \big) 
\end{align}
On the other hand, the left-hand side of (\cref{ineq:7}) can be rewritten as
\begin{align*}
&\E(|X_{i} - X'_{i}| \cdot \mathds{1}_{T \in [X_{i}, X'_{i}]})\\
=\;& \E((X_{i} - X'_{i}) \cdot \mathds{1}_{X_{i} \ge T} \cdot \mathds{1}_{ X'_{i} < T}) + \E((X'_{i} - X_{i}) \cdot \mathds{1}_{X_{i} < T} \cdot \mathds{1}_{ X'_{i} \ge T})
= 2\E((X_{i} - X'_{i}) \cdot \mathds{1}_{X_{i} \ge T} \cdot \mathds{1}_{ X'_{i} < T})
\end{align*}
since $X_{i}$ and $X'_{i}$ are independent random variables with the same distribution. By linearity of expectation we get
\begin{align*}
2\E((X_{i} - X'_{i}) \mathds{1}_{X_{i} \ge T} \cdot \mathds{1}_{ X'_{i} < T}) \le\; &2\big(\E((T - X'_{i}) \mathds{1}_{X_{i} \ge T} \cdot \mathds{1}_{ X'_{i} < T}) + \E((X_{i} - T) \cdot\mathds{1}_{X_{i} \ge T} \cdot \mathds{1}_{ X'_{i} < T})\big)\\
\le\; &2\big(\E((T - X'_{i}) \cdot \mathds{1}_{ X'_{i} < T}) + \E((X_{i} - T) \cdot\mathds{1}_{X_{i} \ge T} ) \big),
\end{align*}
where the last inequality follows from the independence of $X_{i}$ and $X'_{i}$.
This inequality, combined with \cref{eq:6.10} leads to
\begin{align*}
\sum_{j \in [n], j \neq i} 2\E(|X_{i} - X'_{j}| \cdot \mathds{1}_{T \in [X_{i}, X'_{j}]})
\ge \bigg(\frac{n}{16} - \frac{1}{4}\bigg)\cdot \big(\E((T - X_{i}) \cdot \mathds{1}_{X_{i} \le T}) + \E((X_{i} - T) \cdot \mathds{1}_{X_{i} > T}) \big)
\end{align*}
\begin{align*}
\ge \frac{1}{2}\cdot\bigg(\frac{n}{16} - \frac{1}{4}\bigg)\cdot\E(|X_{i} - X'_{i}| \cdot \mathds{1}_{T \in [X_{i}, X'_{i}]}),
\end{align*}
which is equivalent to
$$ \frac{64}{n} \cdot\E(|X_{i} - X'_{j}| \cdot \mathds{1}_{T \in [X_{i}, X'_{j}]}) \ge \E(|X_{i} - X'_{i}| \cdot \mathds{1}_{T \in [X_{i}, X'_{i}]}).$$
This proves the claimed \cref{ineq:7} with constant $C := 64$, and therefore the lemma follows.
\end{proof}

We are now ready to finish the proof of \cref{thm:random-order-v2}.
\begin{proof}[Proof of \cref{thm:random-order-v2}.]
Recall that our goal is to prove
\begin{align}\label{eq:6.6}
\bigg(1 + \frac{C}{n - C} \bigg) \cdot 2 \ge \frac{S_{1}}{S_{2}}.
\end{align}
To do so, first observe that
\begin{align*}
S_{1} = \E(|X_{\sigma(1)} - X_{\sigma(2)}|) = \frac{n}{n - 1}\cdot \E(|X_{Y} - X'_{Y'}|) - \frac{1}{n(n-1)}\cdot\sum_{i \in [n]} \E(|X_{i} - X'_{i}|)
\end{align*}
\begin{align*}
= \frac{n}{n - 1}\cdot S'_{1} -  \frac{1}{n(n-1)}\cdot\sum_{i \in [n]} \E(|X_{i} - X'_{i}|),
\end{align*}
which leads to
\begin{align*}
\frac{n}{n - 1}\cdot S'_{1} = S_{1} + \frac{1}{n(n-1)}\cdot\sum_{i \in [n]} \E(|X_{i} - X'_{i}|).
\end{align*}
As observed before, we have that $2 \ge \frac{S'_{1}}{S'_{2}}$ which combined with the above equality gets that
\begin{align}\label[ineq]{ineq:3}
2 &\ge \frac{ S_{1} + \frac{1}{n(n-1)}\cdot\sum_{i \in [n]} \E(|X_{i} - X'_{i}|)} {S'_{2}} \ge \frac{S_{1}} {S'_{2}},
\end{align}
where in the last inequality we used the fact that $\sum_{i \in [n]} \E(|X_{i} - X'_{i}|) \ge 0$.
By \cref{lem:dem} we get that there exists a constant $C$ such that
$\bigg(1 + \frac{C}{n}\bigg)\cdot S_{2} \ge \frac{n}{n - 1}\cdot S'_{2} \ge S'_{2}$,
which, when plugged into the denominator of the right-hand side of \cref{ineq:3}, implies
$2\cdot \bigg(1 + \frac{C}{n}\bigg)\cdot S_{2} \ge S_{1}$.
This completes the proof of the theorem.
\end{proof}

\section{Generalizations}\label{sec:extensions}

In the following subsections, we consider four generalizations of our basic trading-prophet problem: An unknown-distribution version with sample access, a version with more than one item, a budgeted version with re-investments of gains, and a multi-armed--bandit version.

\subsection{Unknown Distribution and Affiliated Prices}
\label{sec:unknown_dist}
We consider in this section a variant of the model studied in \cref{sec:iid} in which the prices are i.i.d., but we are not given the distribution beforehand.
We prove the result for i.i.d.~prices, but notice that the exact same argument applies to affiliated prices where $p_j = x_j + y$ and $y \sim G$ and $x_j\sim F$ for all $j$ independently. The reason is that once $y$ is fixed the increments are $i.i.d.$, and both the online algorithm and the optimal offline algorithm buy and sell the same number of times so that $y$ cancels out.

\begin{theorem}
\label{thm:unknown_dist}
If prices are i.i.d.\ or affiliated, but we do not know the distribution, there is an algorithm $\ALG$ such that $\E(\ALG)\geq \frac{1}{2}\cdot \frac{n-2}{n-1}\cdot \E(\OPT)$.
\end{theorem}
To prove this result we first analyze the case where, instead of knowing the distribution, we have access to $n-1$ independent samples $S_1,S_2,\ldots,S_{n-1}$. For this case, consider the following algorithm we denote by $\ALG^s$. In period $i<n$, if we do not have the item, we buy if $X_i<S_i$; if we have the item, we sell if $X_i\geq S_i$. In period $n$ we simply sell the item if we still have it. For the case of affiliated prices, we require that the samples are also correlated, i.e., that $S_j=y+s_j$ and $X_j=y+x_j$, where $y\sim G$ and $s_j\sim F$, $x_j\sim F$ for all $j$.

\begin{lemma}
\label{lem:alg_samples}
$\E(\ALG^s)\geq \frac{n-1}{4}\cdot \E(|X_1-X_2|)$.
\end{lemma}
\begin{proof}
Denote by $\ALG^s(i)$ the gains of the algorithm in period $i$. Note first that in every period $i\geq 2$ the event that we have the item is independent of $X_i$ and $S_i$, and has probability exactly $1/2$. This is because this event is exactly the event $\{X_{i-1}<S_{i-1}\}$. Therefore,
\[
\E(\ALG^s(i))=\begin{cases}
  \E(-X_1\cdot \mathds{1}_{X_1<S_1}) &\text{if } i=1\\
  \frac{1}{2}\cdot\E(X_i \cdot \mathds{1}_{X_i\geq S_i} )
  + \frac{1}{2}\cdot\E(-X_i \cdot \mathds{1}_{X_i< S_i})
  &\text{if } 2\leq i\leq n-1\\
  \frac{1}{2} \cdot\E(X_n) &\text{if } i=n.
\end{cases}
\]
Since $X_1,\ldots,X_n,S_1,\ldots S_{n-1}$ are i.i.d. realizations of the same distribution, we conclude that
\begin{align*}
    \E(\ALG^s) = \sum_{i=1}^n \E(\ALG^s(i)) &= \frac{n-1}{2}\cdot\E((X_1-X_2)\cdot \mathds{1}_{X_1\geq X_2})\\
    &= \frac{n-1}{4}\cdot\E(|X_1-X_2|).
\end{align*}
For the case of affiliated random variables, i.e., if $X_j=x_j+y$ and $S_j=s_j+y$ for $y\sim G$ and $x_j\sim s_j\sim F$, notice that the events $\{X_{i-1}<S_{i-1}\}$ and $\{x_{i-1}<s_{i-1}\}$ are equivalent. This means that the event that we have the item in period $i$ is independent of $X_i$ and $S_i$, so all previous equations still hold.
\end{proof}
Combining this lemma with \cref{lem:iid_upper_bound_opt} we have that $\ALG^s$ is actually a $1/2$ approximation. From $\ALG^s$ we derive $\ALG^{s*}$, an algorithm that uses a single sample $S_1$ of the distribution. For $2\leq i\leq n-1$, we define $S'_i$ by selecting a uniformly random element from $\{S_1,X_1,\ldots, X_{i-1}\}$. Now, $\ALG^{s*}$ behaves like $\ALG^s$, taking $S_1,S'_2,\ldots, S'_{n-1}$ as the set of samples.
\begin{lemma}
\label{lem:single_sample_alg}
$\E(\ALG^{s*})=\E(\ALG^s)$.
\end{lemma}
\begin{proof}
We simply prove that the expected gain of the two algorithms is the same in each period, i.e., $\E(\ALG^{s*}(i))=\E(\ALG^s(i))$. This is easy to  see for period $i=1$. For $i>1$, note first that $S'_i$ and $X_i$ are i.i.d. random variables. This is because $\{S_1,X_1,\ldots,X_{i-1}\}$ is independent of $X_i$.

The critical step is to show that the event that we have the item in period $i$ is independent of the pair $(X_i,S'_i)$. Recall that this is equivalent to the event $\{X_{i-1}<S'_{i-1}\}$, so conditioning on this event might affect the distribution of $S'_i$. This in fact is not the case, since it refers to the relative order within the set $\{S_1,X_1,\ldots,X_{i-1}\}$, which is independent of the value of a uniformly random element of the same set (given that they are i.i.d. random variables). Notice that if we now add the same number $y\sim G$ to all random variables, the relative orders do not change, so the argument also applies to affiliated prices.
\end{proof}

Now, \cref{thm:unknown_dist} is a direct application of \cref{lem:single_sample_alg}. If we do not know the distribution, we can  skip the first period and use $\ALG^{s*}$ in $X_2,\ldots,X_n$, treating $X_1$ as a sample. This yields an expected profit of $\frac{n-2}{4}\cdot\E(|X_1-X_2|)$, which  by \cref{lem:iid_upper_bound_opt} is exactly $\frac{1}{2}\cdot \frac{n-2}{n-1} \cdot\E(\OPT)$.

\subsection{More Than One Item}
In this section we analyze a variant of the main model where we are allowed to store up to $k$ copies of the item and the prices are independent and are presented in uniformly random order. It is clear that, in hindsight, the optimal strategy always buy $k$ items or sell all $k$ items. We prove the optimal online strategy also presents this behavior, which implies that this variant is equivalent to the single item case.
\begin{lemma}
\label{lem:more_than_one_item}
If we are allowed to store $k$ items at any time, then the optimal online strategy always trades with all $k$ items, i.e. in a single period either the algorithm buys $k$ items or it sells all $k$ items.  
\end{lemma}
\begin{proof}
Let $V_{t,s}$ be the random variable denoting the gain of the optimal algorithm that plays only in periods from $t$ to $n$ and in the beginning of the period $t$ it has $s$ items and its disposal. Recall, that by optimal strategy we refer to the strategy that maximizes the expectation. For instance, in the above notation we have that
$$\E(\OPT) = E(V_{1, 0}).$$
In the rest part of the proof, we will derive a recursive relation that characterizes all possible moves of any algorithm starting at period $t$ with $s$ items and argue that the expected gains is always maximized at the extremes: moves that either sell all $s$ items or buy the remaining $k - s$ items.

Let us begin with the recurrence for the last period $n$. Since this is the last period, the optimal strategy always sells all items at price $X_{\sigma(n)}$. Conditioned on the random variables $X_{\sigma(1)}, \ldots, X_{\sigma(n - 1)}$ we have that:
\begin{align*}
\E( V_{n, s} \mid  X_{\sigma(1)}, \ldots, X_{\sigma(n - 1)}) &= \int s \cdot X_{\sigma(t)} \;\mathrm{d}P(X_{\sigma(1)}, \ldots, X_{\sigma(n - 1)})\\
&= s \cdot \int X_{\sigma(t)} \;\mathrm{d}P(X_{\sigma(1)}, \ldots, X_{\sigma(n - 1)}),
\end{align*}
where we use $P(X_{\sigma(1)}, \ldots, X_{\sigma(n - 1)})$ to denote the conditional probability distribution of $X_{\sigma(n)}$ given $X_{\sigma(1)}, \ldots, X_{\sigma(n - 1)}$. Let us conclude here that $\E( V_{n, s} \mid  X_{\sigma(1)}, \ldots, X_{\sigma(n - 1)})$ is a linear function of $s$.

For any other period $1 \le t \le n - 1$ the algorithm can buy or sell depending on the number $s$ of items it has. Iterating over all possible choices $s'$ of the number of items the algorithm will hold in the period $t + 1$ we obtain
\begin{align*}
  &\E( V_{t,s} \mid X_{\sigma(1)}, \ldots, X_{\sigma(n - 1)} ) \\
  =\;&\int \max_{0 \le s' \le k}( (s - s') X_{\sigma(t)} + \E(V_{t + 1, s'} \mid X_{\sigma(1)}, \ldots, X_{\sigma(t)}) )  \;\mathrm{d}P(X_{\sigma(1)}, \ldots, X_{\sigma(t - 1)}),  
\end{align*}
where we use $P(X_{\sigma(1)}, \ldots, X_{\sigma(t - 1)})$ to denote the conditional probability distribution of $X_{\sigma(t)}$ given $X_{\sigma(1)}, \ldots, X_{\sigma(t - 1)}$. 
Pulling out the constant $s\cdot X_{\sigma(t)}$ before the max function and by the linearity of the integral we get that
\begin{align}
&\phantom{= =} \int \max_{0 \le s' \le k}( (s - s') X_{\sigma(t)} + \E(V_{t + 1, s'} \mid X_{\sigma(1)}, \ldots, X_{\sigma(t)}) )  \;\mathrm{d}P(X_{\sigma(1)}, \ldots, X_{\sigma(t - 1)}) \notag\\
&= \int s\cdot X_{\sigma(t)} + \max_{0 \le s' \le k}( -s' \cdot X_{\sigma(t)} + \E(V_{t + 1, s'} \mid X_{\sigma(1)}, \ldots, X_{\sigma(t)}) )  \;\mathrm{d}P(X_{\sigma(1)}, \ldots, X_{\sigma(t - 1)}) \notag\\
&= s\int X_{\sigma(t)} \;\mathrm{d}P(X_{\sigma(1)}, \ldots, X_{\sigma(t - 1)})\notag \\
&\quad \quad \qquad + \int \max_{0 \le s' \le k}( -s' \cdot X_{\sigma(t)} + \E(V_{t + 1, s'} \mid X_{\sigma(1)}, \ldots, X_{\sigma(t)}) ) \;\mathrm{d}P(X_{\sigma(1)}, \ldots, X_{\sigma(t - 1)}). \label{eq:8}
\end{align}
We observe that the last expression is a linear function of $s$ assuming that the realizations of $X_{\sigma(1)}, \ldots, X_{\sigma(t - 1)}$  as well as $t$ are given. This combined with the base case $t = n$ proves that for any $t \in [n]$ the conditional expectation
$$\E( V_{t,s} \mid X_{\sigma(1)}, \ldots, X_{\sigma(n - 1)} )$$ is a linear function of $s$.
Applying this observation to \cref{eq:8} we get that, for any $t \in [n]$ the term inside the integral
$$ \int \max_{0 \le s' \le k}( -s' \cdot X_{\sigma(t)} + \E(V_{t + 1, s'} \mid X_{\sigma(1)}, \ldots, X_{\sigma(t)}) )  \;\mathrm{d}P(X_{\sigma(1)}, \ldots, X_{\sigma(t - 1)})$$
is linear in $s'$ and as a such it takes the maximum in either $s' = 0$ or $s' = k$. This proves that the optimal strategy of any algorithm always requires selling all the items or buying the number of items that reaches the limit $k$. Given the fact that the algorithm starts with $0$ items the lemma follows. 
\end{proof}

\subsection{Budgeted Version with Fractional Purchase and Re-investment of Gains}

A natural extension of our model is to allow, as in a stock market, the transaction of any fraction of the item; and at a given period, only limit ourselves by the budget we have in that period. More precisely, in this section we consider the following model. We are given a sequence of positive prices $X_{\sigma(1)},\ldots,X_{\sigma(n)}$ generated in the same way as in our basic model. In period $i$ we have a state given by a pair $(S_i,B_i)$, where $S_i\geq 0$ is the (possibly fractional) number of stocks we hold in period $i$, and $B_i\geq 0$ is our budget in period $i$. We start with $(S_1,B_1)=(0,1)$, meaning that we start with $0$ stocks and a budget of $1$ unit of money. In period $i$ we can buy (or sell) any number of stocks $s\in [-S_i,B_i/X_{\sigma(i)}]$, which determines the state in the next period as $(S_{i+1},B_{i+1})=(S_i+s, B_i - s\cdot X_{\sigma(i)})$. Our objective is to maximize the money we have at the end of the process, i.e., $B_{n+1}$. 

Denote by $\OPT_F$ the gains of the prophet in this model, i.e., the maximum possible profit in hindsight. It is easy to see that the prophet either sells all stocks or spends all money, i.e., $s\in \{-S_i,B_i/X_{\sigma(i)}\}$. Thus, if $I$ denotes the set of periods where the prophet buys, and $J$ the periods where the prophet sells, we have that $\OPT_F=\prod_{j\in J} X_{\sigma(j)}/\prod_{i\in I} X_{\sigma(i)}$. This formula immediately implies that, in this model, the prophet behaves as the prophet of the basic model with prices $X'_{\sigma(i)}=\log(X_{\sigma(i)})$.
In turn, this implies that the prophet has an expected profit that grows exponentially in $n$.
\begin{observation}
By  Jensen's inequality and~\cref{lemma:opt-est-ran-ord}, we have that
\[ \E(\OPT_F)\geq \exp(\E(\log \OPT_F)) = \exp\left(\frac{n-1}{2}\cdot \E(|\log X_{\sigma(1)} - \log X_{\sigma(2)}|)\right). \]
\end{observation}
It is clear from this that there is no hope for a constant approximation against $\E(\OPT_F)$. Instead, we can approximate the expected growth rate of $\OPT_F$, i.e. $\E(\log \OPT_F)$. Indeed, by considering prices $X'_{\sigma(i)}=\log(X_{\sigma(i)})$, we can apply our algorithms from \cref{sec:iid,sec:rand-order} to obtain approximately optimal expected growth rates. In this model, for $T\in \R_+$ we denote by $\ALG_T$ the algorithm that sells all stocks if $X_{\sigma(i)}\geq T$ or $i=n$, and spends the whole budget when $X_{\sigma(i)}<T$ and $i<n$.

\begin{corollary}
If the prices are i.i.d. and $T$ is the median of the distribution, then $\E(\log \ALG_T)\geq \frac{1}{2}\E(\log \OPT_F)$.
\end{corollary}

\begin{corollary}
If the prices are independent and presented in uniformly random order, then there exists a threshold $T$ such that $\E(\log \ALG_T)\geq \frac{1}{16}\E(\log \OPT_F)$.
\end{corollary}

\begin{corollary}
There is a constant $C$ so that if prices are independent and presented in uniformly random order, then there exists a threshold $T$ such that
$\E(\log \ALG_T)\geq \frac{1}{2\cdot (1+C/n)}\E(\log \OPT_F)$.
\end{corollary}

\subsection{Multi-armed--bandit version}
Consider the following model we call $k$-armed bandit trading prophet. We are able to trade $k$ different kinds of items, each with a sequence of prices generated independently as in the basic model, but we can hold at most one item, across all kinds. More precisely, let us denote $F^{i}_{1}, \ldots, F^{i}_{n}$ the initial $n$ distributions of an item $1 \le i \le k$ which are all independent. Let $X^{i}_{j} \sim F^{i}_{j}$ for $1 \le j \le n$ denote $j$-th price of $i$-th item. On period $i$ we observe $k$ different prices $X^1_{\sigma^{1}(i)}, \ldots, X^k_{\sigma^{k}(i)}$ where $\sigma^{j} : \{1, \ldots, n \} \rightarrow \{1, \ldots, n \}$ for $1 \le j \le k$ are $k$ permutations chosen independently from the uniform distribution over all permutations. Like in the basic model, in each period we have two possibilities: Either we have an item $j$ for $1 \le j \le k$, and in this case we can sell this item with the price $X^{j}_{\sigma^{j}(i)}$; or we do not have \textit{any} item, and in this case we can buy any item $j$ for $1 \le j \le k$ for the price $X^{j}_{\sigma^{j}(i)}$. On top of that, we assume that if in a given period we have an item, then we can sell it and within the same period buy a different one.\footnote{Note that this assumption makes no difference in the original model, since it would be equivalent to simply keeping the item until the next period.} 
We start at period $1$ with no item and, as before, we want to design a decision rule that maximizes the expected profit obtained after $n$ periods.

Consider now an algorithm that in the beginning of an execution chooses one item uniformly at random, and next trades only on this item applying the threshold decision rule from \cref{sec:rand-order}. In the following, we will show that this algorithm carries the results stated in \cref{thm:random-order} and \cref{thm:random-order-v2} with $\frac{1}{k}$ multiplicative loss, against the prophet, 
The key observation involves understanding the strategy of the prophet in the multi-armed bandit version.

\begin{lemma}\label{lem:k-bandit-opt}
The expected gain of the optimal strategy in hindsight for the $k$-armed bandit trading-prophet problem is
$$\E(\OPT) = (n-1)\cdot \E\left(\max_{i \in [k]} \left[X^{i}_{\sigma^i(1)} - X^{i}_{\sigma^i(2)}\right]_+ \right), \text{ where } [\cdot]_+:=\max\{\cdot,0\}.$$
\end{lemma}
\begin{proof}
Consider a difference in prices between period $i$ and $i + 1$ for $1 \le i \le n - 1$. On one hand, the prophet can gain at most $\max_{j \in [k]} \big( \big(X^{j}_{\sigma^{j}(i + 1)} - X^{j}_{\sigma^{j}(i)} \big) \mathds{1}_{X^{j}_{\sigma^{j}(i + 1)} > X^{j}_{\sigma^{j}(i)}} \big)$, since he can possess only one item during this period. On the other hand, he can gain exactly this value, since performing this transaction has no effects on next periods as he starts the next period without any item. Thus, we have that
$$\E(\OPT) = \E\left(\sum_{i = 1}^{n - 1} \max_{j \in [k]} 
\left[
X^{j}_{\sigma^{j}(i + 1)} - X^{j}_{\sigma^{j}(i)} \right]_+
\right).$$
Using linearity of the expectation
and the fact that the permutations $\sigma^j$ are independent and uniformly random, we conclude the formula in the statement of the lemma.
\end{proof}

This result allows us to immediately extend \cref{thm:random-order,thm:random-order-v2}.
\begin{theorem}
\label{thm:k_bandit}
There exist algorithms $\ALG_{1}$ and $\ALG_{2}$ that achieve the following approximation for $k$-armed bandit version of trading prophets:
$$\E(\ALG_{1}) \ge \frac{1}{16k} \cdot \E(\OPT);\quad\E(\ALG_{2}) \ge \bigg(\frac{1}{2k} - o(1)\bigg) \cdot \E(\OPT).$$
\end{theorem}
\begin{proof}

Let $\ALG_{1}$ be an algorithm that first picks the one of $k$ items uniformly and random and then applies strategy \cref{thm:random-order} to prices of the chosen item. 
$$\E(\ALG_{1}) = \sum_{i = 1}^{k}\frac{1}{k}\cdot\E(\ALG^{i}_{1}),$$
where $\E(\ALG^{i}_{1})$ denotes the expected gain of algorithm given by \cref{thm:random-order} applied to prices of $i$-th item.
By \cref{lemma:alg-est-ran-ord}, we have that
\begin{align}\label[ineq]{line:10}
\E(\ALG^{i}_{1}) \ge \frac{1}{16}\cdot\E(\OPT^{i}),
\end{align}
where $\OPT^{i}$ denotes the expectation of the optimal strategy in hindsight applied to prices of $i$-th item. By \cref{lemma:opt-est-ran-ord}, we get
$$\frac{1}{16}\cdot\E(\OPT^{i}) = \frac{1}{16}\frac{n - 1}{2}\cdot\E(|X^{i}_{\sigma^{i}(1)} - X^{i}_{\sigma^{i}(2)}|) = \frac{n - 1}{16}\cdot\E\Big(\big[X^{i}_{\sigma^{i}(1)} - X^{i}_{\sigma^{i}(2)}\big]_{+}\Big). $$
The above inequality plugged into \cref{line:10} yields
$$\E(\ALG_{1}) = \sum_{i = 1}^{k}\frac{1}{k}\cdot\E(\ALG^{i}_{1}) \ge \frac{1}{16k}\cdot\sum_{i=1}^{k}(n - 1)\cdot\E\Big(\big[X^{i}_{\sigma^{i}(1)} - X^{i}_{\sigma^{i}(2)}\big]_{+}\Big).$$
Observe that all random variables inside expectations are non-negative. Thus,
$$\sum_{i=1}^{k}(n - 1)\cdot\E\Big(\big[X^{i}_{\sigma^{i}(1)} - X^{i}_{\sigma^{i}(2)}\big]_{+}\Big) \ge (n-1)\cdot \E\left(\max_{i \in [k]} \left[X^{i}_{\sigma^i(1)} - X^{i}_{\sigma^i(2)}\right]_+ \right) =  \E(\OPT),$$
where the last equality holds by \cref{lem:k-bandit-opt}. Therefore the first part of the theorem is proven. 

To prove the second part it suffices to take an algorithm that settles the threshold with respect to \cref{thm:random-order-v2} instead of \cref{thm:random-order}. Since \cref{lemma:opt-est-ran-ord} and \cref{lem:k-bandit-opt} hold irrespectively of the threshold chosen by the algorithm, thus other parts of the above reasoning are valid and proves that  $\ALG_{2}$ has an approximation factor of $\frac{1}{2k} - o(1)$.
\end{proof}

The natural question is whether there is any algorithm with an approximation factor $o(1/k)$. Not surprisingly, the answer is negative even we restrict to i.i.d. random variables $X_{j}^{i}$.
\begin{lemma}
\label{lem:k_bandit_instance}
For $k\geq 2$, there is an instance of the $k$-armed bandit trading-prophet problem with i.i.d. prices $X_{j}^{i}$ such that for any algorithm $\ALG$ it holds that
$$\E(\ALG) \le \frac{1}{\big(1 - e^{-1/2}\big)k}\cdot \E(\OPT).$$
\end{lemma}
\begin{proof}
Consider a random variable $X$ that takes $0$ with probability $1 - \frac{1}{k}$ and $k$ with probability $\frac{1}{k}$. Assume that $n = 2$ and each random variable $X_{j}^{i}$ for $1 \le i \le k$, $1 \le j \le n$ is drawn independently according to the distribution $X$. Since there are only two periods, there can be at most one transaction. Obviously, the prophet trades only if the price on the first period is $0$ and the price of the same item on the second period is $k$. For a fixed item the probability of such prices realization is $\big(1 - \frac{1}{k}\big)\frac{1}{k}$. Since we have $k$ different items whose prices drawn independently, we see that the probability the sequence of prices $(0, k)$ does not happen in any of these prices realizations is
$$\bigg(1 - \Big(1 - \frac{1}{k}\Big)\frac{1}{k}\bigg)^{k} = \bigg(1 - \frac{1}{k} + \frac{1}{k^2}\bigg)^{k} \le \bigg(1 - \frac{1}{2k}\bigg)^{k} \le \sqrt{\frac{1}{e}}.$$
Therefore, we have
$$\E(\OPT) \ge \big(1-e^{-1/2}\big)\cdot k.$$
On the other hand, since any algorithm can make only one transaction, it always buys an item that have price $0$ in the first period. Since prices presented in the second period are independent from choices of the algorithm, it does not matter which item the algorithm selects among these with price $0$. Therefore, we have that
$$\E(\ALG) \le k \cdot \frac{1}{k} = 1,$$
which, if compared with $\E(\OPT)$, proves the lemma.
\end{proof}

\bibliographystyle{plainnat}
\bibliography{bibliography}

\begin{thebibliography}{34}
\providecommand{\natexlab}[1]{#1}
\providecommand{\url}[1]{\texttt{#1}}
\expandafter\ifx\csname urlstyle\endcsname\relax
  \providecommand{\doi}[1]{doi: #1}\else
  \providecommand{\doi}{doi: \begingroup \urlstyle{rm}\Url}\fi

\bibitem[Abolhassani et~al.(2017)Abolhassani, Ehsani, Esfandiari, Hajiaghayi,
  Kleinberg, and Lucier]{AbolhassaniEEHK17}
Melika Abolhassani, Soheil Ehsani, Hossein Esfandiari, MohammadTaghi
  Hajiaghayi, Robert~D. Kleinberg, and Brendan Lucier.
\newblock Beating $1-1/e$ for ordered prophets.
\newblock In \emph{{ACM} {SIGACT} Symposium on Theory of Computing (STOC)},
  pages 61--71, 2017.

\bibitem[Alaei(2014)]{Alaei14}
Saeed Alaei.
\newblock Bayesian combinatorial auctions: Expanding single buyer mechanisms to
  many buyers.
\newblock \emph{{SIAM} Journal on Computing}, 43\penalty0 (2):\penalty0
  930--972, 2014.

\bibitem[Borodin and El-Yaniv(1998)]{BorodinEY98}
Allan Borodin and Ran El-Yaniv.
\newblock \emph{Online Computation and Competitive Analysis}.
\newblock Cambridge University Press, 1998.

\bibitem[Braun and Kesselheim(2021)]{BraunK21}
Alexander Braun and Thomas Kesselheim.
\newblock Truthful mechanisms for two-sided markets via prophet inequalities.
\newblock In \emph{{ACM} Conference on Economics and Computation (EC)}, pages
  202--203, 2021.

\bibitem[Brustle et~al.(2017)Brustle, Cai, Wu, and Zhao]{BrustleCWZ17}
Johannes Brustle, Yang Cai, Fa~Wu, and Mingfei Zhao.
\newblock Approximating gains from trade in two-sided markets via simple
  mechanisms.
\newblock In \emph{{ACM} Conference on Economics and Computation (EC)}, pages
  589--590, 2017.

\bibitem[Caramanis et~al.(2022)Caramanis, D{\"{u}}tting, Faw, Fusco, Lazos,
  Leonardi, Papadigenopoulos, Pountourakis, and
  Reiffenh{\"{a}}user]{CaramanisEtAl22}
Constantine Caramanis, Paul D{\"{u}}tting, Matthew Faw, Federico Fusco, Philip
  Lazos, Stefano Leonardi, Orestis Papadigenopoulos, Emmanouil Pountourakis,
  and Rebecca Reiffenh{\"{a}}user.
\newblock Single-sample prophet inequalities via greedy-ordered selection.
\newblock In \emph{{ACM-SIAM} Symposium on Discrete Algorithms (SODA)}, pages
  1298--1325, 2022.

\bibitem[Charnes et~al.(1966)Charnes, Dr\`eze, and Miller]{CharnesEtAl66}
Abraham Charnes, Jacques Dr\`eze, and Merton Miller.
\newblock Decision and horizon rules for stochastic planning problems: A linear
  example.
\newblock \emph{Econometrica}, 34:\penalty0 307--330, 1966.

\bibitem[Colini{-}Baldeschi et~al.(2016)Colini{-}Baldeschi, de~Keijzer,
  Leonardi, and Turchetta]{Colini-Baldeschi16}
Riccardo Colini{-}Baldeschi, Bart de~Keijzer, Stefano Leonardi, and Stefano
  Turchetta.
\newblock Approximately efficient double auctions with strong budget balance.
\newblock In \emph{{ACM-SIAM} Symposium on Discrete Algorithms (SODA)}, pages
  1424--1443, 2016.

\bibitem[Colini{-}Baldeschi et~al.(2017{\natexlab{a}})Colini{-}Baldeschi,
  Goldberg, de~Keijzer, Leonardi, Roughgarden, and
  Turchetta]{Colini-Baldeschi17}
Riccardo Colini{-}Baldeschi, Paul~W. Goldberg, Bart de~Keijzer, Stefano
  Leonardi, Tim Roughgarden, and Stefano Turchetta.
\newblock Approximately efficient two-sided combinatorial auctions.
\newblock In \emph{{ACM} Conference on Economics and Computation (EC)}, pages
  591--608, 2017{\natexlab{a}}.

\bibitem[Colini{-}Baldeschi et~al.(2017{\natexlab{b}})Colini{-}Baldeschi,
  Goldberg, de~Keijzer, Leonardi, and Turchetta]{Colini-Baldeschi17a}
Riccardo Colini{-}Baldeschi, Paul~W. Goldberg, Bart de~Keijzer, Stefano
  Leonardi, and Stefano Turchetta.
\newblock Fixed price approximability of the optimal gain from trade.
\newblock In \emph{International Conference on Web and Internet Economics
  (WINE)}, pages 146--160, 2017{\natexlab{b}}.

\bibitem[Correa et~al.(2019)Correa, D{\"{u}}tting, Fischer, and
  Schewior]{CorreaDFS19}
Jos{\'{e}}~R. Correa, Paul D{\"{u}}tting, Felix~A. Fischer, and Kevin Schewior.
\newblock Prophet inequalities for {I.I.D.} random variables from an unknown
  distribution.
\newblock In \emph{{ACM} Conference on Economics and Computation (EC)}, pages
  3--17, 2019.

\bibitem[Correa et~al.(2020)Correa, Cristi, Epstein, and Soto]{CorreaCES20}
Jos{\'{e}}~R. Correa, Andr{\'{e}}s Cristi, Boris Epstein, and Jos{\'{e}}~A.
  Soto.
\newblock The two-sided game of googol and sample-based prophet inequalities.
\newblock In \emph{{ACM-SIAM} Symposium on Discrete Algorithms (SODA)}, pages
  2066--2081, 2020.

\bibitem[Correa et~al.(2021{\natexlab{a}})Correa, D{\"{u}}tting, Fischer,
  Schewior, and Ziliotto]{CorreaDFSZ21}
Jos{\'{e}}~R. Correa, Paul D{\"{u}}tting, Felix~A. Fischer, Kevin Schewior, and
  Bruno Ziliotto.
\newblock Unknown {I.I.D.} prophets: Better bounds, streaming algorithms and a
  new impossibility.
\newblock In \emph{Innovations in Theoretical Computer Science Conference
  ({ITCS})}, pages 86:1--86:1, 2021{\natexlab{a}}.

\bibitem[Correa et~al.(2021{\natexlab{b}})Correa, Foncea, Hoeksma, Oosterwijk,
  and Vredeveld]{CorreaFHOV21}
Jos{\'{e}}~R. Correa, Patricio Foncea, Ruben Hoeksma, Tim Oosterwijk, and Tjark
  Vredeveld.
\newblock Posted price mechanisms and optimal threshold strategies for random
  arrivals.
\newblock \emph{Mathematics of Operations Research}, 46\penalty0 (4):\penalty0
  1452--1478, 2021{\natexlab{b}}.

\bibitem[Correa et~al.(2021{\natexlab{c}})Correa, Saona, and
  Ziliotto]{CorreaSZ21}
Jos{\'{e}}~R. Correa, Raimundo Saona, and Bruno Ziliotto.
\newblock Prophet secretary through blind strategies.
\newblock \emph{Mathematical Programming}, 190\penalty0 (1):\penalty0 483--521,
  2021{\natexlab{c}}.

\bibitem[Correa et~al.(2022)Correa, Hartline, and Immorlica]{CorreaIH2022}
Jose~R. Correa, Jason Hartline, and Nicole Immorlica.
\newblock A semester virtual institute.
\newblock
  \url{https://cacm.acm.org/blogs/blog-cacm/258538-a-semester-virtual-institute/fulltext},
  2022.

\bibitem[Deng et~al.(2022)Deng, Mao, Sivan, and Wang]{DengSW21}
Yuan Deng, Jieming Mao, Balasubramanian Sivan, and Kangning Wang.
\newblock Approximately efficient bilateral trade.
\newblock In \emph{{ACM} {SIGACT} Symposium on Theory of Computing (STOC)},
  pages 718--721, 2022.

\bibitem[du~Toit and Peskir(2007)]{DuToitP07}
Jacques du~Toit and Goran Peskir.
\newblock The trap of complacency in predicting the maximum.
\newblock \emph{The Annals of Probability}, 35:\penalty0 340--365, 2007.

\bibitem[du~Toit and Peskir(2009)]{DuToitP09}
Jacques du~Toit and Goran Peskir.
\newblock Selling a stock at the ultimate maximum.
\newblock \emph{The Annals of Applied Probability}, 19:\penalty0 983--1014,
  2009.

\bibitem[D{\"{u}}tting et~al.(2021)D{\"{u}}tting, Fusco, Lazos, Leonardi, and
  Reiffenh{\"{a}}user]{DuttingFLLR21}
Paul D{\"{u}}tting, Federico Fusco, Philip Lazos, Stefano Leonardi, and Rebecca
  Reiffenh{\"{a}}user.
\newblock Efficient two-sided markets with limited information.
\newblock In \emph{{ACM} {SIGACT} Symposium on Theory of Computing (STOC)},
  pages 1452--1465, 2021.

\bibitem[Ehsani et~al.(2018)Ehsani, Hajiaghayi, Kesselheim, and
  Singla]{EhsaniHKS18}
Soheil Ehsani, MohammadTaghi Hajiaghayi, Thomas Kesselheim, and Sahil Singla.
\newblock Prophet secretary for combinatorial auctions and matroids.
\newblock In \emph{{ACM-SIAM} Symposium on Discrete Algorithms (SODA)}, pages
  700--714, 2018.

\bibitem[Esfandiari et~al.(2017)Esfandiari, Hajiaghayi, Liaghat, and
  Monemizadeh]{EsfandiariHLM17}
Hossein Esfandiari, MohammadTaghi Hajiaghayi, Vahid Liaghat, and Morteza
  Monemizadeh.
\newblock Prophet secretary.
\newblock \emph{{SIAM} Journal on Discrete Mathematics}, 31\penalty0
  (3):\penalty0 1685--1701, 2017.

\bibitem[Graversen et~al.(2006)Graversen, Peskir, , and
  Shiryaev]{GraversenEtAl06}
Svend~E. Graversen, Goran Peskir, , and Albert~N. Shiryaev.
\newblock Stopping brownian motion without anticipation as close as possible to
  its ultimate maximum.
\newblock \emph{Theory of Probability and its Applications}, 45:\penalty0
  41--50, 2006.

\bibitem[Kleinberg and Kleinberg(2018)]{KleinbergK18}
Jon~M. Kleinberg and Robert Kleinberg.
\newblock Delegated search approximates efficient search.
\newblock In \emph{{ACM} Conference on Economics and Computation (EC)}, pages
  287--302, 2018.

\bibitem[Krengel and Sucheston(1977)]{KrengelS77}
Ulrich Krengel and Louis Sucheston.
\newblock Semiamarts and finite values.
\newblock \emph{Bulletin of the American Mathematical Society}, 83:\penalty0
  745--747, 1977.

\bibitem[Krengel and Sucheston(1978)]{KrengelS78}
Ulrich Krengel and Louis Sucheston.
\newblock On semiamarts, amarts, and processes with finite value.
\newblock \emph{Advances in Probability and Related Topics}, 4:\penalty0
  197--266, 1978.

\bibitem[Li and Hoi(2014)]{LiHoi14}
Bin Li and Steven~C. Hoi.
\newblock Online portfolio selection: {A} survey.
\newblock \emph{ACM Computing Surveys}, 46:\penalty0 1--33, 2014.

\bibitem[Liu et~al.(2021)Liu, Leme, P{\'{a}}l, Schneider, and Sivan]{LiuLPSS21}
Allen Liu, Renato~Paes Leme, Martin P{\'{a}}l, Jon Schneider, and
  Balasubramanian Sivan.
\newblock Variable decomposition for prophet inequalities and optimal ordering.
\newblock In \emph{{ACM} Conference on Economics and Computation (EC)}, page
  692, 2021.

\bibitem[McAfee(2008)]{McAfee08}
R.~Preston McAfee.
\newblock The gains from trade under fixed price mechanisms.
\newblock \emph{Applied Economics Research Bulletin}, 1:\penalty0 1--10, 2008.

\bibitem[Myerson and Satterthwaite(1983)]{MyersonS83}
Roger~B. Myerson and Mark~A. Satterthwaite.
\newblock Efficient mechanisms for bilateral trading.
\newblock \emph{Journal of Economic Theory}, 29:\penalty0 265--281, 1983.

\bibitem[Osborne(1959)]{Osborne59}
Matthew F.~M. Osborne.
\newblock Brownian motion in the stock market.
\newblock \emph{Operations Research}, 7\penalty0 (2):\penalty0 145--173, 1959.

\bibitem[Rubinstein et~al.(2020)Rubinstein, Wang, and Weinberg]{RubinsteinWW20}
Aviad Rubinstein, Jack~Z. Wang, and S.~Matthew Weinberg.
\newblock Optimal single-choice prophet inequalities from samples.
\newblock In \emph{Innovations in Theoretical Computer Science Conference
  (ITCS)}, pages 60:1--60:10, 2020.

\bibitem[Samuel-Cahn(1984)]{SamuelC84}
Esther Samuel-Cahn.
\newblock Comparison of threshold stop rules and maximum for independent
  nonnegative random variables.
\newblock \emph{Annals of Probability}, 12:\penalty0 1213--1216, 1984.

\bibitem[Singla(2018)]{Singla18}
Sahil Singla.
\newblock \emph{Combinatorial Optimization Under Uncertainty: Probing and
  Stopping-Time Algorithms}.
\newblock PhD thesis, Carnegie Mellon University, 2018.

\end{thebibliography}

\newpage
\appendix

\section{Omitted Proofs and Additional Material}

\subsection{Proof of \cref{prop:iid_lb}}\label{appx:subsec:iid_lb}

\begin{proof}
Fix $n$, we construct a parameterized instance with $n$ i.i.d.~prices and obtain the claimed bound as the parameter $\varepsilon \in [0,1]$ tends to $0$.

Fix $\varepsilon \in [0,1]$ and consider
\[
X_i = \begin{cases}
1 & \text{w.p.~$\frac{\varepsilon}{2}$}\\
\frac{1}{2} & \text{w.p.~$1 - \varepsilon$}\\
0 &\text{w.p.~$\frac{\varepsilon}{2}$}
\end{cases} \qquad \text{for $i \in \{1,2, \ldots, n\}$.}
\]
To calculate the expected profit of $\OPT$, we use the characterization described in Lemma~\ref{lem:iid_upper_bound_opt}. Here, it gives 
\begin{align*}
\E(\OPT) &= \frac{n-1}{2}\cdot\E\left(|X_{1} - X_{2}|\right)\\
&= \frac{n-1}{2}\cdot\left(4\cdot\frac{1}{2}\cdot(1-\varepsilon)\cdot \frac{\varepsilon}{2} + 2\cdot 1 \cdot \left(\frac{\varepsilon}{2}\right)^2 \right) = \frac{n-1}{2}\cdot\left(\varepsilon - \frac{\varepsilon^2}{2} \right).
\end{align*}

Consider any algorithm $\ALG$. We can bound its expected profit as follows. First, selling is always optimal in the last period. Since the expected price in this period is $\frac{1}{2}$, we can replace it with a deterministic price of $\frac{1}{2}$, which by linearity of expectation does not change the profit of $\ALG$. Second, any optimal algorithm performs buy/sell operations only when the prices belong to one of the following pairs: $\left(\frac{1}{2}, \frac{1}{2} \right)$, $\left(0, \frac{1}{2} \right)$, $\left(0, 1\right)$, $\left(\frac{1}{2}, 1\right)$. In all these cases, we can distribute the profit from the buy/sell operations into two parts. The algorithm gains $0$ when it trades at a value of $\frac{1}{2}$, and it gains $\frac{1}{2}$ each time it trades at $0$ (buy) or $1$ (sell). Now, consider period $t$, and let $q(t)$ represent the probability that $\ALG$ holds an item before observing the price in this period. Since the price at period $t$ is independent of $q(t)$, thus using the above reasoning, for $t\leq n-1$ we can write that
\[
\E(\text{profit of }\ALG\text{ at period } t) \leq q(t)\frac{1}{2}\varepsilon + (1-q(t))\frac{1}{2}\varepsilon = \frac{1}{4}\varepsilon.
\]
By the linearity of expectation, we conclude that
\[
\E(\ALG) \le \frac{(n-1)}{4}\cdot\varepsilon.
\]
It follows that
\[
\frac{\E(\ALG)}{\E(\OPT)} \le \frac{\frac{n-1}{4}\cdot\varepsilon}{\frac{n-1}{2}\left(\varepsilon - \frac{\varepsilon^2}{2}\right)}.
\]
Sending $\varepsilon$ to $0$, we derive that 
\[
\lim_{\varepsilon \rightarrow 0} \frac{\E(\ALG)}{\E(\OPT)} \le \lim_{\varepsilon \rightarrow 0}  \frac{\frac{n-1}{4}\cdot\varepsilon}{\frac{n-1}{2}\left(\varepsilon - \frac{\varepsilon^2}{2}\right)} = \frac{1}{2}. \qedhere
\]

\end{proof}

\subsection{Proof of \cref{prop:rdm_order_lb}}\label{appx:rdm_order_lb}

\begin{proof}
Consider a large constant $M>0$. We define an instance with $n=2$. Take $X_1$ such that $X_1=M+2$ w.p. $M/(M+2)$, and $X_1=0$ w.p. $2/(M+2)$. Take $X_2$ such that $X_2=M$ w.p. $M/(M+2)$ and $X_2=2M+2$ w.p. $2/(M+2)$.

Notice that $\E(X_1)=M$ and $\E(X_2)=M+2$. The expectation of the optimal algorithm can be written as
\begin{align*}
    \E(\ALG)&= \frac{1}{2}\cdot \E([\E(X_2)-X_1]_+) + \frac{1}{2}\cdot\E([\E(X_1)-X_2]_+)\\
    &= \frac{1}{2}\cdot (M+2)\cdot \frac{2}{M+2} +0 = 1.
\end{align*}

Let us calculate $\E(\OPT)$. There are four possible sequences of prices where $\OPT$ buys and sells: $(M+2,2M+2), (0,M), (0,2M+2), (M,M+2)$. In all other cases $\OPT$ does nothing. Therefore,
\begin{align*}
    &\E(\OPT)\\
    &= \frac{1}{2}\cdot\left( M\cdot \frac{M}{M+2}\cdot \frac{2}{M+2} + M\cdot \frac{2}{M+2}\cdot \frac{M}{M+2} + (2M+2)\cdot  \left(\frac{2}{M+2}\right)^2 + 2\cdot  \left(\frac{M}{M+2}\right)^2 \right)\\
    &= 3 + O(1/M).
\end{align*}
This completes the proof.
\end{proof}

\subsection{Counterexample for Median of Mixture as Threshold}\label{sec:hard-median}
In this section, we consider the algorithm that sets the median of the mixture distribution as a threshold. We present two distributions of prices that, when presented in random order, show that the gap between the value achieved by this algorithm and that of the prophet can be arbitrarily large. For an arbitrary $\epsilon > 0$, consider the following distributions:
\begin{equation*}
X_{1} = \begin{cases}
             0  & \text{w.\ p.\ } 2\epsilon, \\
             1  & \text{w.\ p.\ } (1 - 2\epsilon),
       \end{cases} \quad
X_{2} = \begin{cases}
             0  & \text{w.\ p.\ } (1 - \epsilon), \\
             \frac{1}{\epsilon^{2}}  & \text{w.\ p.\ } \epsilon.
       \end{cases}
\end{equation*}
First, note that the median of the mixture distribution of $X_{1}$ and $X_{2}$ is $0$. Next, assume that an algorithm sets threshold to $T := 0$. In consequence the algorithm buys only when it sees price $0$ on the first period, which leads to
$$\E(\ALG_{T}) \leq \frac{1}{2}(1-2\epsilon)(1-\epsilon)(1 - 0) + \frac{1}{2}\cdot2\epsilon^{2}\bigg(\frac{1}{\epsilon^{2}} - 0 \bigg) \le \frac{3}{2}.$$
On the other hand,
$$\E(\OPT) = \frac{1}{2}(1-2\epsilon)(1-\epsilon)(1 - 0) + \frac{1}{2}\cdot2\epsilon^{2}\bigg(\frac{1}{\epsilon^{2}} - 0 \bigg) + \frac{1}{2}(1-2\epsilon)\epsilon\bigg(\frac{1}{\epsilon^{2}} - 1\bigg) = \Omega\bigg(\frac{1}{\epsilon}\bigg).$$
This shows that the threshold algorithm that uses the median of the mixture distribution as a threshold does not achieve a constant-factor approximation. 

\subsection{Proof of \cref{lemma:two-medians}}\label{appx:subsec:two-medians}
\begin{proof}
Let $M_{1}$ be the median of the first price distribution $F_{1}$ and let $M_{2}$ be the median of the second price distribution $F_{2}$. Let us assume without loss of generality that $M_{1} \le M_{2}$. 

We show that for at least one value from the set $\{M_{1}, M_{2}\}$ setting the threshold $T$ to this value satisfies \cref{eq:1}. 
In order to show that the bigger of the two expectations corresponding to $T$ being $M_{1}$ or $M_{2}$ is greater than $\frac{1}{4}\E(|X_{1} - X_{2}|)$ it suffices to show the following inequality
\begin{align}\label[ineq]{eq:2}
& 2\left( \E(|X_{1} - X_{2}| \cdot \mathds{1}_{M_{1} \in [\min(X_{1}, X_{2}), \max(X_{1}, X_{2})]}) + \E(|X_{1} - X_{2}| \cdot \mathds{1}_{M_{2} \in [\min(X_{1}, X_{2}), \max(X_{1}, X_{2})]}) \right) \notag\\ &\hspace*{48pt}\ge \E(|X_{1} - X_{2}|)
\end{align}

Let us first rewrite the left-hand side of \cref{eq:2}) using our assumption that $M_1 = \min\{M_1,M_2\}$ and $M_2 = \max\{M_1,M_2\}$ as follows
\begin{align}
&2\left( \E\big(|X_{1} - X_{2}| \cdot \mathds{1}_{M_{1} \in [\min(X_{1}, X_{2}), \max(X_{1}, X_{2})]}\big) + \E\big(|X_{1} - X_{2}| \cdot  \mathds{1}_{M_{2} \in [\min(X_{1}, X_{2}), \max(X_{1}, X_{2})]}\big)\right) \notag\\
&\quad= 2\bigg( \E\big((X_{2} - X_{1}) \cdot \mathds{1}_{X_{1} \le M_{1}} \cdot \mathds{1}_{X_{2} \ge M_{1}} \big) +  \E\big((X_{1} - X_{2}) \cdot \mathds{1}_{X_{1} \ge M_{1}} \cdot \mathds{1}_{X_{2} \le M_{1}}\big) \bigg) \notag\\
&\quad\qquad + 2\bigg( \E\big((X_{2} - X_{1})\cdot \mathds{1}_{X_{1} \le M_{2}} \cdot \mathds{1}_{X_{2} \ge M_{2}}\big) +  \E\big((X_{1} - X_{2})\mathds{1}_{X_{1} \ge M_{2}} \cdot \mathds{1}_{X_{2} \le M_{2}}\big) \bigg) \notag\\
&\quad= 4\bigg(\E\big((X_{2} - X_{1}) \cdot \mathds{1}_{X_1 \le M_{1}} \cdot \mathds{1}_{X_2 \ge M_{2}} \big) + \E\big((X_{1} - X_{2}) \cdot \mathds{1}_{X_1 \ge M_{2}} \cdot \mathds{1}_{X_2 \le M_{1}}\big) \bigg)\notag\\
&\quad\qquad+ 2\bigg( \E\big((X_{2} - X_{1}) \cdot \mathds{1}_{X_1 \le M_{1}} \cdot \mathds{1}_{X_2 \in [M_{1}, M_{2}]}\big) + \E\big((X_{2} - X_{1}) \cdot \mathds{1}_{X_1 \in [M_{1}, M_{2}]} \cdot \mathds{1}_{X_2 \ge M_{2}}\big)\notag\\
&\quad\qquad\qquad+ \E\big((X_{1} - X_{2}) \cdot \mathds{1}_{X_1 \ge M_{2}} \cdot \mathds{1}_{X_2 \in [M_{1}, M_{2}]} \big)+ \E\big((X_{1} - X_{2}) \cdot \mathds{1}_{X_1 \in [M_{1}, M_{2}]} \cdot \mathds{1}_{X_2 \le M_{1}}\big) \bigg). \label{eq:lhs-expanded}  
\end{align}

Using again that $M_1 = \min\{M_1,M_2\}$ and $M_2 = \max\{M_1,M_2\}$, the right-hand side of \cref{eq:2} can be expanded to
\begin{align}
\E(|X_{1} - X_{2}|) &=\phantom{+}\E(|X_{2} - X_{1}|\cdot \mathds{1}_{X_{1} \le M_1}\cdot \mathds{1}_{X_{2} \le M_1}) + \E((X_{2} - X_{1})\cdot \mathds{1}_{X_{1} \le M_1}\cdot \mathds{1}_{X_{2} \in [M_{1}, M_2]})& \notag\\ &\hspace{1cm}+\E((X_{2} - X_{1})\cdot \mathds{1}_{X_{1} \le M_1}\cdot \mathds{1}_{X_{2} \ge M_2}) \notag\\
&\hspace*{11pt}+\E((X_{1} - X_{2})\cdot \mathds{1}_{X_{1} \in [M_1, M_2]} \cdot \mathds{1}_{X_{2} \le M_1})  + \E(|X_{2} - X_{1}| \cdot \mathds{1}_{X_{1} \in [M_1, M_2]} \cdot \mathds{1}_{X_{2} \in [M_{1}, M_2]}) & \notag\\ &\hspace{1cm}+\E((X_{2} - X_{1})\cdot \mathds{1}_{X_{1} \in [M_1, M_2]} \cdot \mathds{1}_{X_{2} \ge M_2}) \notag\\
&\hspace*{11pt}+\E((X_{1} - X_{2}) \cdot \mathds{1}_{X_{1} \ge M_2} \cdot \mathds{1}_{X_{2} \le M_1})  + \E((X_{1} - X_{2}) \cdot \mathds{1}_{X_{1} \ge M_2} \cdot \mathds{1}_{X_{2} \in [M_{1}, M_2]}) & \notag\\ &\hspace{1cm}+\E(|X_{1} - X_{2}| \cdot \mathds{1}_{X_{1} \ge M_2}\cdot \mathds{1}_{X_{2} \ge M_2}). \label{eq:rhs-expanded}
\end{align}

Next, we will show how to regroup components of the sum in \cref{eq:lhs-expanded} corresponding to the left-hand side of~\cref{eq:2} to get an upper bound of the components of the sum in \cref{eq:rhs-expanded} corresponding to the right-hand side of inequality \cref{eq:2}. 

\medskip

We will use:

\begin{observation}\label{obs:technical}
If $X_1 \sim F_1$, $X_2 \sim F_2$ and $M_{1} \leq M_2$ are the medians of $F_1$ and $F_2$, then the following inequalities hold:
\begin{enumerate}
    \item $\E\big((X_{2} - X_1) \cdot \mathds{1}_{X_1 \le M_1} \cdot \mathds{1}_{X_2 \ge M_1}\big) + \E\big((X_{1} - X_2) \cdot \mathds{1}_{X_1 \ge M_1} \cdot \mathds{1}_{X_2 \le M_1}\big)  \ge \E\big(|X_{2} - X_{1}| \cdot \mathds{1}_{X_{1} \le M_1} \cdot \mathds{1}_{X_{2} \le M_1}\big)$
    \item $\E\big((X_{2} - X_{1}) \cdot \mathds{1}_{X_1 \le M_1} \cdot \mathds{1}_{X_2 \ge M_2}\big) + \E\big((X_{1} - X_{2})\cdot \mathds{1}_{X_1 \ge M_2} \cdot \mathds{1}_{X_2 \le M_1}\big) \ge \E\big(|X_{2} - X_{1}| \cdot \mathds{1}_{X_{1} \in [M_1, M_2]} \cdot \mathds{1}_{X_{2} \in [M_{1}, M_2]}\big)$
    \item $\E\big((X_{2} - X_{1}) \cdot \mathds{1}_{X_1 \le M_2} \cdot \mathds{1}_{X_2 \ge M_2})\big) + \E\big((X_{1} - X_{2}) \cdot \mathds{1}_{X_1 \ge M_2} \cdot \mathds{1}_{X_2 \le M_2})\big) \ge \E\big(|X_{1} - X_{2}| \cdot \mathds{1}_{X_{1} \ge M_2} \cdot \mathds{1}_{X_{2} \ge M_2}\big)$
    
\end{enumerate}
\end{observation}

\begin{proof}
\begin{enumerate}
    \item Since $M_{2}$ is the median of $F_{2}$ and $M_{1} \le M_{2}$ we have that 
    $$\PP(X_{2} \ge M_{1}) \ge \PP(X_{2} < M_{1}).$$
    Therefore, we have that
    $$\E\big((X_{2} - X_1) \cdot \mathds{1}_{X_1 \le M_1} \cdot \mathds{1}_{X_2 \ge M_1}\big) \ge \E\big((X_2 - X_1) \cdot \mathds{1}_{X_1 \le M_1} \cdot \mathds{1}_{X_2 \le M_1} \cdot \mathds{1}_{X_2 > X_1}\big).$$
    Similarly, because $\PP(X_{1} \le M_{1}) = \PP(X_{1} > M_{1})$ and we decrease values of $X_{1}$ on the right-hand side, we get that
    $$\E\big((X_{1} - X_2)\cdot \mathds{1}_{X_1 \ge M_1} \cdot \mathds{1}_{X_2 \le M_1}\big) \ge \E((X_1 - X_2) \cdot \mathds{1}_{X_1 \le M_1} \cdot \mathds{1}_{X_2 \le M_1}\cdot\mathds1_{X_1 \ge X_2}).$$
    Since
    \begin{align*}
    &\E\big((X_2 - X_1) \cdot \mathds{1}_{X_1 \le M_1} \cdot \mathds{1}_{X_2 \le M_1} \cdot \mathds1_{X_2 > X_1}\big) + \E((X_1 - X_2) \cdot \mathds{1}_{X_1 \le M_1} \cdot \mathds{1}_{X_2 \le M_1}\cdot\mathds1_{X_1 \ge X_2}) \\
    &\qquad= \E\big(|X_{2} - X_{1}| \cdot \mathds{1}_{X_{1} \le M_1} \cdot \mathds{1}_{X_{2} \le M_1}\big),
    \end{align*}
    thus $1)$ must hold. 
    \item 
    Observe, that the inequality: 
    $$(X_2 - X_1) \cdot \mathds{1}_{X_1 \le M_1} \cdot \mathds{1}_{X_2 \ge M_2} \ge |X_{2} - X_{1}| \cdot \mathds{1}_{X_{1} \in [M_1, M_2]} \cdot \mathds{1}_{X_{2} \in [M_{1}, M_2]}$$ 
    holds with probability $1$.
    Also, because $M_{1}, M_{2}$ are medians we have that
    $$\PP\big(\mathds{1}_{X_1 \le M_1} \cdot \mathds{1}_{X_2 \ge M_2}\big) \ge \frac{1}{4} \ge \PP\big(\mathds{1}_{X_{1} \in [M_1, M_2]}\cdot \mathds{1}_{X_{2} \in [M_{1}, M_2]}\big),$$
    which together gives that
    $$\E\big((X_{2} - X_{1}) \cdot \mathds{1}_{X_1 \le M_1} \cdot \mathds{1}_{X_2 \ge M_2}\big) \ge \E\big(|X_{2} - X_{1}| \cdot \mathds{1}_{X_{1} \in [M_1, M_2]} \cdot \mathds{1}_{X_{2} \in [M_{1}, M_2]}\big).$$
    Observe, that the random variable $(X_{1} - X_{2})\cdot \mathds{1}_{X_1 \ge M_2} \cdot \mathds{1}_{X_2 \le M_1}$ is non-negative, since we assumed $M_{1} \le M_{2}$. Combining this with the previous inequality leads to
    \begin{align*}
        &\E\big((X_{2} - X_{1}) \cdot \mathds{1}_{X_1 \le M_1} \cdot \mathds{1}_{X_2 \ge M_2}\big) + \E\big((X_{1} - X_{2})\cdot \mathds{1}_{X_1 \ge M_2} \cdot \mathds{1}_{X_2 \le M_1}\big)\\
        &\qquad\geq\E\big(|X_{2} - X_{1}| \cdot \mathds{1}_{X_{1} \in [M_1, M_2]} \cdot \mathds{1}_{X_{2} \in [M_{1}, M_2]}\big) ,
    \end{align*}
    which proves $2)$.
    \item 
    Again, since $M_2$ is the median of $F_{2}$ and $M_{1} \le M_{2}$ we have that
    $$\PP(X_1 \le M_{2}) \ge \PP(X_{1} \ge M_{2}),$$
    which implies the following inequality
    $$\E\big((X_{2} - X_{1}) \cdot \mathds{1}_{X_1 \le M_2} \cdot \mathds{1}_{X_2 \ge M_2})\big) \ge \E\big((X_{2} - X_{1}) \cdot \mathds{1}_{X_{1} \ge M_2} \cdot \mathds{1}_{X_{2} \ge M_2} \cdot \mathds{1}_{X_{2} \ge X_1} \big).$$
    It also holds that
    $$\E\big((X_{1} - X_{2}) \cdot \mathds{1}_{X_1 \ge M_2} \cdot \mathds{1}_{X_2 \le M_2})\big) \ge
    \E\big((X_{1} - X_{2}) \cdot \mathds{1}_{X_{1} \ge M_2} \cdot \mathds{1}_{X_{2} \ge M_2} \cdot \mathds{1}_{X_{1} \ge X_2} \big),$$
    since $\PP(X_{2} \le M_{2}) = \PP(X_{2} > M_{2})$ and the left-hand side increases values of $X_{2}$.
    Since
    \begin{align*}
        &\E\big((X_{2} - X_{1}) \cdot \mathds{1}_{X_{1} \ge M_2} \cdot \mathds{1}_{X_{2} \ge M_2} \cdot \mathds{1}_{X_{2} \ge X_1} \big) + \E\big((X_{1} - X_{2}) \cdot \mathds{1}_{X_{1} \ge M_2} \cdot \mathds{1}_{X_{2} \ge M_2} \cdot \mathds{1}_{X_{1} \ge X_2} \big)\\
        &\qquad=\E\big(|X_{1} - X_{2}| \cdot \mathds{1}_{X_{1} \ge M_2} \cdot \mathds{1}_{X_{2} \ge M_2}\big),
    \end{align*}
    adding both sides of the two above inequalities we get $3)$.
\end{enumerate}
This finishes the proof of the observation.
\end{proof}

Recall, that we wanted to compare \cref{eq:lhs-expanded} with \cref{eq:rhs-expanded}. By adding inequalities in \cref{obs:technical} side by side we get that
\begin{align}
&\E\big((X_{2} - X_1) \cdot \mathds{1}_{X_1 \le M_1} \cdot \mathds{1}_{X_2 \ge M_1}\big) + \E\big((X_{1} - X_2) \cdot \mathds{1}_{X_1 \ge M_1} \cdot \mathds{1}_{X_2 \le M_1}\big) \notag\\
+\;& \E\big((X_{2} - X_{1}) \cdot \mathds{1}_{X_1 \le M_1} \cdot \mathds{1}_{X_2 \ge M_2}\big) + \E\big((X_{1} - X_{2})\cdot \mathds{1}_{X_1 \ge M_2} \cdot \mathds{1}_{X_2 \le M_1}\big) \notag\\
+\;&\E\big((X_{2} - X_{1}) \cdot \mathds{1}_{X_1 \le M_2} \cdot \mathds{1}_{X_2 \ge M_2})\big) + \E\big((X_{1} - X_{2}) \cdot \mathds{1}_{X_1 \ge M_2} \cdot \mathds{1}_{X_2 \le M_2})\big) \notag\\
\ge\hspace*{11pt}&\E\big(|X_{2} - X_{1}| \cdot \mathds{1}_{X_{1} \le M_1} \cdot \mathds{1}_{X_{2} \le M_1}\big) \notag\\
+\;& \E\big(|X_{2} - X_{1}| \cdot \mathds{1}_{X_{1} \in [M_1, M_2]} \cdot \mathds{1}_{X_{2} \in [M_{1}, M_2]}\big) \notag\\ 
+\;&\E\big(|X_{1} - X_{2}| \cdot \mathds{1}_{X_{1} \ge M_2} \cdot \mathds{1}_{X_{2} \ge M_2}\big). \label[ineq]{eq:7}  
\end{align}
Note that
\begin{align*}
&\E\big((X_{2} - X_1) \cdot \mathds{1}_{X_1 \le M_1} \cdot \mathds{1}_{X_2 \ge M_1}\big) + \E\big((X_{1} - X_2) \cdot \mathds{1}_{X_1 \ge M_1} \cdot \mathds{1}_{X_2 \le M_1}\big) \\
=\hspace*{11pt}&\E\big((X_{2} - X_1) \cdot \mathds{1}_{X_1 \le M_1} \cdot \mathds{1}_{X_2 \in [M_1, M_2]}\big) + \E\big((X_{2} - X_1) \cdot \mathds{1}_{X_1 \le M_1} \cdot \mathds{1}_{X_2 \ge M_2}\big) \\
+\;&\E\big((X_{1} - X_2) \cdot \mathds{1}_{X_1 \in [M_1, M_2]} \cdot \mathds{1}_{X_2 \le M_1}\big) + \E\big((X_{1} - X_2) \cdot \mathds{1}_{X_1 \ge M_2} \cdot \mathds{1}_{X_2 \le M_1}\big)
\end{align*}
and similarly
\begin{align*}
&\E\big((X_{2} - X_{1}) \cdot \mathds{1}_{X_1 \le M_2} \cdot \mathds{1}_{X_2 \ge M_2})\big) + \E\big((X_{1} - X_{2}) \cdot \mathds{1}_{X_1 \ge M_2} \cdot \mathds{1}_{X_2 \le M_2})\big) \\
=\hspace*{11pt}& \E\big((X_{2} - X_{1}) \cdot \mathds{1}_{X_1 \le M_1} \cdot \mathds{1}_{X_2 \ge M_2})\big) + \E\big((X_{2} - X_{1}) \cdot \mathds{1}_{X_1 \in [M_1, M_2]} \cdot \mathds{1}_{X_2 \ge M_2})\big) \\
+\;& \E\big((X_{1} - X_{2}) \cdot \mathds{1}_{X_1 \ge M_2} \cdot \mathds{1}_{X_2 \le M_1})\big) + \E\big((X_{1} - X_{2}) \cdot \mathds{1}_{X_1 \ge M_2} \cdot \mathds{1}_{X_2 \in [M_1, M_2]})\big).
\end{align*}
Therefore \cref{eq:7} can be rewritten as
\begin{align*}
&3\bigg(\E\big((X_{2} - X_{1}) \cdot \mathds{1}_{X_1 \le M_{1}} \cdot \mathds{1}_{X_2 \ge M_{2}} \big) + \E\big((X_{1} - X_{2}) \cdot \mathds{1}_{X_1 \ge M_{2}} \cdot \mathds{1}_{X_2 \le M_{1}}\big) \bigg)\notag\\
+\;& \E\big((X_{2} - X_{1}) \cdot \mathds{1}_{X_1 \le M_{1}} \cdot \mathds{1}_{X_2 \in [M_{1}, M_{2}]}\big) + \E\big((X_{2} - X_{1}) \cdot \mathds{1}_{X_1 \in [M_{1}, M_{2}]} \cdot \mathds{1}_{X_2 \ge M_{2}}\big)\notag\\
+\;& \E\big((X_{1} - X_{2}) \cdot \mathds{1}_{X_1 \ge M_{2}} \cdot \mathds{1}_{X_2 \in [M_{1}, M_{2}]} \big)+ \E\big((X_{1} - X_{2}) \cdot \mathds{1}_{X_1 \in [M_{1}, M_{2}]} \cdot \mathds{1}_{X_2 \le M_{1}}\big) \notag\\
\ge\hspace*{11pt}&\E\big(|X_{2} - X_{1}| \cdot \mathds{1}_{X_{1} \le M_1} \cdot \mathds{1}_{X_{2} \le M_1}\big) \\
+\;& \E\big(|X_{2} - X_{1}| \cdot \mathds{1}_{X_{1} \in [M_1, M_2]} \cdot \mathds{1}_{X_{2} \in [M_{1}, M_2]}\big) \notag\\
+\;&\E\big(|X_{1} - X_{2}| \cdot \mathds{1}_{X_{1} \ge M_2} \cdot \mathds{1}_{X_{2} \ge M_2}\big).
\end{align*}
By adding 
\begin{align*}
&\bigg(\E\big((X_{2} - X_{1}) \cdot \mathds{1}_{X_1 \le M_{1}} \cdot \mathds{1}_{X_2 \ge M_{2}} \big) + \E\big((X_{1} - X_{2}) \cdot \mathds{1}_{X_1 \ge M_{2}} \cdot \mathds{1}_{X_2 \le M_{1}}\big) \bigg)\notag\\
+\;&\E\big((X_{2} - X_{1}) \cdot \mathds{1}_{X_1 \le M_{1}} \cdot \mathds{1}_{X_2 \in [M_{1}, M_{2}]}\big) + \E\big((X_{2} - X_{1}) \cdot \mathds{1}_{X_1 \in [M_{1}, M_{2}]} \cdot \mathds{1}_{X_2 \ge M_{2}}\big)\notag\\
+\;& \E\big((X_{1} - X_{2}) \cdot \mathds{1}_{X_1 \ge M_{2}} \cdot \mathds{1}_{X_2 \in [M_{1}, M_{2}]} \big)+ \E\big((X_{1} - X_{2}) \cdot \mathds{1}_{X_1 \in [M_{1}, M_{2}]} \cdot \mathds{1}_{X_2 \le M_{1}}\big),
\end{align*}
to both sides of the above inequality we get that
\begin{align*}
&4\bigg(\E\big((X_{2} - X_{1}) \cdot \mathds{1}_{X_1 \le M_{1}} \cdot \mathds{1}_{X_2 \ge M_{2}} \big) + \E\big((X_{1} - X_{2}) \cdot \mathds{1}_{X_1 \ge M_{2}} \cdot \mathds{1}_{X_2 \le M_{1}}\big) \bigg)\notag\\
+\;& 2\bigg( \E\big((X_{2} - X_{1}) \cdot \mathds{1}_{X_1 \le M_{1}} \cdot \mathds{1}_{X_2 \in [M_{1}, M_{2}]}\big) + \E\big((X_{2} - X_{1}) \cdot \mathds{1}_{X_1 \in [M_{1}, M_{2}]} \cdot \mathds{1}_{X_2 \ge M_{2}}\big)\notag\\
&\hspace*{2pt}\;+\E\big((X_{1} - X_{2}) \cdot \mathds{1}_{X_1 \ge M_{2}} \cdot \mathds{1}_{X_2 \in [M_{1}, M_{2}]} \big)+ \E\big((X_{1} - X_{2}) \cdot \mathds{1}_{X_1 \in [M_{1}, M_{2}]} \cdot \mathds{1}_{X_2 \le M_{1}}\big) \bigg)\\
\ge\hspace*{11pt}&\E(|X_{2} - X_{1}|\cdot \mathds{1}_{X_{1} \le M_1}\cdot \mathds{1}_{X_{2} \le M_1}) + \E((X_{2} - X_{1})\cdot \mathds{1}_{X_{1} \le M_1}\cdot \mathds{1}_{X_{2} \in [M_{1}, M_2]})& \notag\\ 
&+\;\E((X_{2} - X_{1})\cdot \mathds{1}_{X_{1} \le M_1}\cdot \mathds{1}_{X_{2} \ge M_2}) \notag\\
+\;&\E((X_{1} - X_{2})\cdot \mathds{1}_{X_{1} \in [M_1, M_2]} \cdot \mathds{1}_{X_{2} \le M_1}) + \E(|X_{2} - X_{1}| \cdot \mathds{1}_{X_{1} \in [M_1, M_2]} \cdot \mathds{1}_{X_{2} \in [M_{1}, M_2]}) & \notag\\ 
&+ \E((X_{2} - X_{1})\cdot \mathds{1}_{X_{1} \in [M_1, M_2]} \cdot \mathds{1}_{X_{2} \ge M_2}) \notag\\
+\;&\E((X_{1} - X_{2}) \cdot \mathds{1}_{X_{1} \ge M_2} \cdot \mathds{1}_{X_{2} \le M_1}) + \E((X_{1} - X_{2}) \cdot \mathds{1}_{X_{1} \ge M_2} \cdot \mathds{1}_{X_{2} \in [M_{1}, M_2]}) & \notag\\
&+ \E(|X_{1} - X_{2}| \cdot \mathds{1}_{X_{1} \ge M_2}\cdot \mathds{1}_{X_{2} \ge M_2}), \end{align*}
which exactly compares values of \cref{eq:lhs-expanded} and \cref{eq:rhs-expanded}.
\end{proof}

\subsection{Tightness of~\cref{lemma:two-medians}}

\begin{observation}
There is an instance for which the factor $\frac{1}{4}$ in~\cref{lemma:two-medians} is tight when  $T$ is the best of the two medians.
\end{observation}

\begin{proof}[Proof sketch.]
Consider the following price distributions:
\begin{equation*}
    X_{1} =
    \begin{cases}
      5 + \varepsilon_{1} + \varepsilon_{2} & \text{with prob. }  \frac{1}{2} - \gamma,  \\
      5 + \varepsilon_{1} & \text{with prob. } \gamma,  \\
      5 & \text{with prob. } \frac{1}{2} - \gamma, \\
      0 & \text{with prob. } \gamma; 
    \end{cases}\\
\end{equation*}
\begin{equation*}
    X_{2} =
    \begin{cases}
      10 & \text{with prob. }  \gamma,  \\
      5 & \text{with prob. } \frac{1}{2} - \gamma, \\
      5 - \varepsilon_{1} & \text{with prob. } \gamma,\\
      5 - \varepsilon_{1} - \varepsilon_{2}  & \text{with prob. } \frac{1}{2} - \gamma.
    \end{cases}\\
\end{equation*}
Assume that $\gamma \sim \frac{1}{100}$ and $\varepsilon_{1}, \varepsilon_{2}$ are much smaller than $\gamma$. We see that $\E(|X_{1} - X_{2}|) \sim 2\cdot 5\gamma = \frac{1}{10} = \frac{20}{200}$.

Now, consider $T$ being the median of $X_{1}$ (the case when $T$ is the median of $X_{2}$ is symmetric). $T$ is  $[5, 5 + \varepsilon_{1}]$. For this threshold we have that
\begin{align*}
\E( |X_{1} - X_{2}| \cdot\mathds{1}_{T \in [\min(X_{1}, X_{2}), \max(X_{1}, X_{2})]} ) &\sim (\varepsilon_{1} + \varepsilon_{2}) \frac{1}{2} + (5 + \varepsilon_{1})\gamma\cdot\gamma + 5(\frac{1}{2} - \gamma)\gamma + 10\gamma^{2}\\
&\sim 0 + \frac{5}{10000} + \frac{5}{200} + \frac{1}{1000} ~ \sim \frac{5}{200},
\end{align*}
completing the proof.
\end{proof}

\end{document}